\documentclass[journal,twoside,print]{ieeecolor}
\usepackage{generic}
\usepackage{cite}
\usepackage{color}
\usepackage{multirow}
\usepackage{booktabs}
\usepackage{amsmath,amssymb,amsfonts}
\usepackage{graphicx}
\usepackage{mathrsfs}                
\usepackage{bm}
\newtheorem{asm}{\bf Assumption}
\newtheorem{lemt}{\bf Lemma}
\newtheorem{thmt}{\bf Theorem}
\newtheorem{rmk}{\bf Remark}
\usepackage{arydshln}
\usepackage{soul}
\definecolor{PPY}{rgb}{0.89,0.4,0.09} 
\definecolor{R}{rgb}{1, 0.325,0.286}

\def\BibTeX{{\rm B\kern-.05em{\sc i\kern-.025em b}\kern-.08em
		T\kern-.1667em\lower.7ex\hbox{E}\kern-.125emX}}
\begin{document}
	\title{Adaptive Distributed Observer-based  Model \\  Predictive Control for Multi-agent Formation \\with Resilience to Communication Link Faults}
	\author{Binyan Xu, \IEEEmembership{Member, IEEE}, Yufan Dai, \IEEEmembership{Graduate Student Member, IEEE}, \\
		Afzal Suleman, and Yang Shi \IEEEmembership{Fellow, IEEE}
		\thanks{This paper was supported by the Natural Sciences and Engineering Research Council of Canada (NSERC).}
		\thanks{B. Xu is with the School of Engineering, University of Guelph, Guelph ON N1G 2W1, Canada (e-mail; binyan@uoguelph.ca)  ;Y. Dai, A. Suleman and Y. Shi are with the Department of Mechanical Engineering, University of Victoria, Victoria BC V8P 5C2, Canada (e-mail: yufandai@uvic.ca; suleman@uvic.ca; yshi@uvic.ca). }\vspace{-1cm}}
	\maketitle
	
	\begin{abstract}
		In order to address the nonlinear multi-agent formation tracking control problem with input constraints and unknown communication faults, a novel adaptive distributed observer-based distributed model predictive control method is developed in this paper. 
		This design employs adaptive distributed observers in local control systems to estimate the leader's state, dynamics, and relative positioning with respect to the leader. Utilizing the estimated data as local references, the original formation tracking control problem can be decomposed into several fully localized tracking control problems, which can be efficiently solved by the local predictive controller. Through the incorporation of adaptive distributed observers, this proposed design not only enhances the resilience of distributed formation tracking against communication faults but also simplifies the distributed model predictive control formulation. 
	\end{abstract}
	
	
	\section{Introduction}
	Multi-agent systems (MASs), distinguished by decentralized task allocation, distributed mission execution, and self-organization, have drawn increasing attention across diverse fields due to their broad spectrum of applications. Moreover, the study of formation control, which involves controlling the positions and orientations of agents to attain a particular geometric configuration with respect to a leader, has emerged as a prominent research topic \cite{ren2011distributed}. While numerous studies have been made in the field of formation control \cite{qin2016recent}, it is worth noting that many investigations do not account for input constraints and control optimality. However, the incorporation of input constraints is imperative for an accurate formulation of real-world problems. The inclusion of optimality is also essential to fully exploit available control resources while satisfying input constraints. 
	
	An appealing framework for formation control is distributed model predictive control (DMPC). DMPC inherits the advantages of centralized model predictive control (MPC), including systematic handling of hard constraints, optimized control performance, inherent robustness, and the ability to cope with nonlinear multi-variable systems \cite{shi2021advanced}. In addition, the distributed implementation fashion of DMPC effectively distributes the computation workload, further enhancing its appeal and practicality \cite{keviczky2004study}. Numerous DMPC methods for MASs have been proposed, as summarized in review papers such as \cite{christofides2013distributed, negenborn2014distributed}. However, existing DMPC results may encounter limitations when tackling the distinctive challenges posed by multi-UAV formation problems. First of all, the computation resources of onboard microcontrollers are limited, while the proposed DMPC methods with terminal constraints demand a sufficiently long prediction horizon and, thereby, a large computation amount to ensure feasibility. Secondly, the majority of DMPC methods are tailored to address the cooperative regulation problem that drives all agents toward a prior-known set point \cite{dunbar2006distributed,li2013robust}. These methods underlie an implicit assumption regarding the communication graph that each agent in the system is directly linked to the leader. Such an assumption is not true in the context of formation control, where the leader’s information is often only available to a portion of the followers. A notable exception, proposed in \cite{zheng2016distributed}, does not require globally known leader information but is only applicable to multi-vehicle platoon scenarios for tracking a constant-speed leader. There also exists a contradiction between the high communication workload of DMPC and the limited bandwidth of wireless communication networks employed in multi-UAV systems. To attain global stability and feasibility of local optimization problems, it is essential for each distributed optimizer to have access to its neighbors' most up-to-date optimized control sequences over the prediction horizon \cite{dunbar2006distributed}. As a result, DMPC usually entails a substantial amount of information exchange, iteratively \cite{mercangoz2007distributed,stewart2010cooperative,venkat2005stability} or sequentially \cite{richards2004decentralized,richards2007robust}. Moreover, the distributed predictive controller, which involves predictions for both itself and its neighbors, not only needs to receive information from its neighbors but also needs to identify the source of that information. 
	
	A significant challenge that remains inadequately addressed in current DMPC studies is maintaining control performance in the presence of communication faults. The communication network plays a crucial role in enabling interactions and facilitating cooperative behaviors among agents. However, the inclusion of communication networks introduces additional vulnerabilities to the control system, particularly when facing unexpected events such as cyber-attacks and channel fading. Communication faults between agents can pose major threats to multi-agent control systems, potentially deteriorating control performance or even overall system stability. Attacks and fading within communication networks can be modeled as uncertainties in the communication links.  Recent studies in \cite{li2019robust,li2014multi,ma2015mean} explore the consensus of MASs with stochastic uncertain communication networks. In \cite{wang2010consensus,zelazo2015robustness}, deterministic network uncertainties are examined within the context of MASs with single integrator agents. In \cite{chen2020adaptive}, a distributed state observer-based adaptive control protocol is designed to address the leader-follower consensus for linear MASs with communication link faults. This study demonstrates that the distributed leader state observer network is resilient to communication link faults. However, it requires that all following agents know the leader dynamics. As an extension of this result,  \cite{yang2021adaptive} proposes adaptive distributed leader state/dynamics observers and control protocols, offering a completely distributed solution for synchronizing linear MASs with time-varying edge weights without the need for global knowledge of the leader dynamics.  Most existing research on resilience cooperative control in the presence of communication uncertainties is directed towards unconstrained MASs with linear dynamics. 
	Moreover, to the best of our knowledge, fully distributed control for formation tracking under communication link faults has not yet received significant research attention. 
	
	Motivated by the aforementioned investigations, this paper develops a novel adaptive distributed observer-based DMPC method for nonlinear MASs in the presence of input constraints and communication link faults. To achieve the formation tracking objective without relying on global access to the leader's information, adaptive distributed observers are developed for all local control systems, estimating online the leader's state, dynamics, and the desired relative position with respect to the leader. With information estimated by these observers, distributed MPC controllers are independently developed to manipulate each agent toward a predetermined formation relative to the estimated leader while adhering to input constraints. The asymptotic convergence of the observation process is demonstrated, which in turn proves the closed-loop control performance of the overall system.  To validate the efficacy of the proposed design, simulations are conducted using both a numerical example and a practical 5-UAV system. The key contributions of this research work include:
	\begin{itemize}
		\item In contrast to prior works such as \cite{wang2010consensus, zelazo2015robustness, chen2020adaptive, yang2021adaptive} that focus on the consensus problem in unconstrained, linear MASs, this study explores the formation tracking control problem in MASs with both input constraints and nonlinear dynamics. 
		Adaptive distributed observers are utilized not only to estimate the leader's state and dynamics but also the desired formation displacements of each agent relative to the leader. With the estimated real-time information as the reference, MPC is employed for the local controller design to achieve optimized control performance subject to the input constraints.
		\item By locally estimating tracking references through corresponding adaptive observers, the distributed formation tracking control task can be decoupled into several fully distributed tracking control problems at the local level. This facilitates the development of local controllers. Therefore, the integration of adaptive observers can significantly reduce the complexity of the distributed MPC formulation compared to other designs proposed in \cite{dunbar2006distributed,dunbar2007distributed,wei2024robust,wei2021robust,wei2019distributed}.
	\end{itemize}
	
	The remainder of this paper is structured as follows: Section \ref{6s: Problem Formulation} provides the mathematical formulation of the control problem and objective; Section \ref{6s: Control Design} elaborates on the distributed control design, presenting the adaptive observer and the MPC-based controller; Section \ref{6s: Stability Analysis} conducts the closed-loop analysis, evaluating the convergence of the estimation and the stability of the control system;  Section \ref{6s: Simulation Study} offers two simulation examples to validate the effectiveness of the proposed design; Finally, Section \ref{6s: Conclusions} summarizes this paper.
	
	Notations used in this paper are listed as follows. $\mathbb{R}$ and $\mathbb{R}^+$ denote the set of rational numbers and the set of positive rational numbers, respectively. $\mathbb{R}^n$ is the set of $n$-dimensional real column vectors, while $\mathbb{R}^{n\times m}$ is the set of $n \times m$ real matrices.  $\bm{x}^\top$ represents the transpose of $\bm{x}$.  $\|\bm{x}\|$ represents the standard Euclidean norm of $\bm{x}$ and $\|\bm{x}\|^2_{\bm{Q}}=\bm{x}^\top\!\bm{Q}\bm{x}$ is the weighed squared norm of $\bm{x}$. $\overline{\sigma}(\bm{A})$ and $\underline{\sigma}(\bm{A})$ represent the minimal and
	maximal eigenvalues of matrix $\bm{A}$, respectively. A diagonal matrix with $x_1, x_2,\cdots x_n$ being the diagonal elements is denoted by $\text{diag}(x_1, x_2,\cdots x_n)$, while a diagonal matrix with the elements of vector $\bm{x}$ on the diagonal is denoted by $\text{diag}(\bm{x})$. A diagonal matrix whose diagonal contains blocks of matrices $\bm{A}_1$,$\bm{A}_2$,$\cdots$,$\bm{A}_n$ is denoted by $\text{blkdiag}(\bm{A}_1,\bm{A}_2,\cdots,\bm{A}_n)$. The notation $\otimes$ is used to denote the Kronecker product. 
	
	\section{Problem Formulation}\label{6s: Problem Formulation}
	This section outlines the mathematical formulation of the control problem to tackle: multi-agent formation tracking control in the presence of communication faults. Firstly, we detail the dynamics models of individual agents and the virtual leader and describe their intercommunication through a directed weighted graph. Subsequently, the modeling of communication faults is presented. Finally, we introduce leader-follower tracking errors to formulate the control objective for formation tracking.

	\subsection{Multi-agent System}
	
	Consider a multi-agent system comprising $M$ followers and one virtual leader. The dynamics of both followers and the leader are detailed below, while their interactions are modeled using a weighted directed graph.

	\subsubsection{Follower Dynamics} The dynamics of the $i$th follower can be described by the following higher-order MIMO nonlinear model:
	\begin{align}\label{6eq: nonlinear systems}
		\left\{
		\begin{array}{rl}
			\dot{x}_{i,1}&\!\!=x_{i,2}\\
			&\vdots\\
			\dot{x}_{i,r\!-\!1}&\!\!=x_{i,r}\\
			\dot{x}_{i,r}&\!\!=f_i(x_i)+G_i(x_i)u_i\\
			y_i&\!\!=x_{i,1}
		\end{array}	\right.
	\end{align}
	where $x_i=\left[{x}_{i,1}^{\top}\ {x}_{i,2}^{\top}\ \cdots\ {x}_{i,r}^{\top}\right]^\top\in\mathbb{R}^{rn}$ is the system state vector with each segment $x_i^l\in\mathbb{R}^{n}$ for $l=1,2,\cdots,r$; $u_i=\left[u_{i,1}\ u_{i,2}\ \cdots\  u_{i,n}\right]\in\mathbb{R}^{n}$ and $y_i=\left[y_{i,1}\ y_{i,2}\ \cdots\  y_{i,n}\right]\in\mathbb{R}^{n}$ are the control input and system output, respectively; $f_i(x_i)=\left[f_{i,1}(x_i)\ f_{i,1}(x_i) \ \cdots\ \right.$ $\left. f_{i,n}(x_i)\right]^\top: \mathbb{R}^{rn}\rightarrow\mathbb{R}^n$ is a vector function, and   $G_i(x_i)=\left[g_{i,1}(x_i)\ g_{i,2}(x_i)\ \cdots g_{i,n}(x_i)\right]:\mathbb{R}^{rn}\rightarrow\mathbb{R}^{n\times n}$ is a square matrix function. To ensure that the system's behavior is predictable and well-behave around the origin, the following assumption is necessary and commonly employed
	\begin{asm}\label{6asm: Lipschitiz continuity}
		Al entries of $f_i(x_i)$ and $G_i(x_i)$ are sufficiently smooth and locally Lipschitz function of $x_i$ and satisfy $f_i(0)=0$ and $G_i(0)\neq0$. 
	\end{asm}
	
	\subsubsection{Communication Graph}The communication among these $M$ followers can be described using a directed weighted graph. Such graph can be represented by $\mathcal{G}=\{\mathcal{V},\mathcal{E}\}$. In this representation, $\mathcal{V}=\{1,2,\cdots,M\}$ is the set of nodes, with each node corresponding to a follower agent. $\mathcal{E}=\{(j,i)|i,j\in\mathcal{V},i\neq j\}$ is the set of edges and $(j,i)\in\mathcal{E}$ means there is a communication link from agent $j$ to agent $i$. Associated with this graph are two critical matrices. The adjacency matrix
	$\mathcal{A}=[a_{ij}]$ is defined such that $a_{ij}>0$ if $(j,i)\in\mathcal{E}$ and $a_{ij}=0$ otherwise. The Laplacian matrix $\mathcal{L}=[l_{ij}]$ is defined with  $l_{ii}=\sum_{j=1}^Ma_{ij}$ capturing the in-degree of node $i $ and $l_{ij}=-a_{ij}$ for $i\neq j$. 
	
	\subsubsection{Leader Dynamics and Connectivity}
	In addition to the follower agents, the system includes a virtual leader whose role is to guide the overall behavior of the MAS. The dynamics of this virtual leader can be governed by:
	\begin{align}\label{6eq: leader dynamics}
		\dot{\xi}_0=S_0\xi_0
	\end{align}
	where  $\xi_0=\left[\xi_{0,1}^{\top}\ {\xi}_{0,2}^{\top}\ \cdots\ \xi_{0,r}^{\top}\right]^\top\in\mathbb{R}^{rn}$ represents the state vector of the leader; $S_0\in\mathcal{R}^{rn\times rn}$ denotes the system dynamics matrix. 
	\begin{rmk}
		It is imperative that the leader’s state vector $\xi_0$ is of equivalent dimensionality to the followers' dynamics, ensuring it can serve as a reference for the followers' outputs. For instance, the $l$th segment of $\xi_0$ serves as the reference for $y_i^{(l)}$ of follower $i$.
	\end{rmk}

	Note that the state vector $\xi_0$ and the dynamics matrix $S_0$ of the leader are only accessible to certain followers. The leader can be labeled as node $0$, and the connections from this leader to the followers, labeled $1, 2, \cdots, M$, can be defined by a set of pinning edges $\mathcal{E}^0 = \{(0, i) | i \in \mathcal{V}\}$. An edge $(0, i) \in \mathcal{E}^0$ indicates that follower $i$ has direct access to the leader’s state and dynamics. Additionally, we introduce a pinning matrix $\mathcal{B} = {\rm diag}(b_1, b_2, \cdots, b_M)$, where $b_i > 0$ if $(0, i) \in \mathcal{E}^0$ and $b_i = 0$ otherwise. This matrix effectively quantifies the influence of the leader on each follower. 
	
	\subsection{Input Constraints and Communication Faults}
	In this work, we address both the input constraints of individual follower agents and unknown faults that may occur within the communication network. These considerations are crucial for ensuring the robustness and reliability of the system under various operational conditions. The mathematical models that incorporate these elements are provided below. 
	
	\subsubsection{Input Constraints}Considering the limitations on excitable control actions, the control input of the $i$th follower is restricted to a nonempty compact convex set, as defined by
	\begin{align}\label{6eq: input constraint}
		u_i\in\Omega_{u_i}\triangleq\left\{u_i\in\mathbb{R}^n|\ -\underline{u}_i\leqslant u_i\leqslant \overline{u}_i\right\}
	\end{align}
	where $\underline{u}_i=[\underline{u}_{i,1}\ \underline{u}_{i,2}\ \cdots\ \underline{u}_{i,n}]^\top\in\mathbb{R}^{n+}$ and $\overline{u}_i=[\overline{u}_{i,1}\ \overline{u}_{i,2}\ \cdots\ \overline{u}_{i,n}]^\top\in\mathbb{R}^{n+}$.  
	
	\subsubsection{Communication Faults}Communication faults can be modeled as time-varying uncertainties affecting the graph edges ~\cite{chen2020adaptive}:
	\begin{subequations}\label{6eq: communication fault}
		\begin{align}
			a^f_{ij}(t)&=a_{ij}+\vartheta^a_{ij}(t)\\
			b^f_{i}(t)&=b_i+\vartheta^b_{i}(t)
		\end{align}
	\end{subequations}
	where $a_{ij}$ and $b_i$ are the idea weights of general and pinning edges, and $\vartheta^a_{ij}$ and $\vartheta^b_{i}$ denote corrupted weights caused by communication faults. This fault model covers the following types of communication faults: 
	\begin{itemize}
		\item {Channel manipulation attack: }The unknown corrupted weights $\vartheta^a_{ij}$ and $\vartheta^b_{i}$ are capable of modeling a range of cyber attacks. The unknown corrupted weights $\vartheta^a_{ij}$ and $\vartheta^b_{i}$ can simulate various cyber attacks. These attacks involve infiltrating communication channels and manipulating the shared information between vehicles.
		\item {Fading channel: }The corrupted communication weights can also represent the effect of a fading channel, resulting in a decrease in the values of the communication weights.
	\end{itemize}
	
	\begin{rmk}
		The existence of $\vartheta^a_{ij}$ and $\vartheta^b_{i}$ introduces time-variation and uncertainty into the weights of the communication links. Consequently, in the event of communication link failures as specified in (\ref{6eq: communication fault}), both the Laplacian matrix and the pinning matrix of the directed graph $\mathcal{G}$ undergo modifications. Specifically, the Laplacian matrix is redefined as $\mathcal{L}^f(t) = [l^f_{ij}(t)]$, where $l^f_{ii}(t) = \sum_{j=1}^M a^f_{ij}(t)$ for the diagonal elements, and $l^f_{ij}(t) = -a^f_{ij}(t)$ for off-diagonal elements with $i \neq j$. Similarly, the pinning matrix is revised to $\mathcal{B}^f(t) = \text{diag}(b_1^f(t), b_2^f(t), \dots, b_M^f(t))$. 
	\end{rmk}
	
	\begin{asm}\label{6asm: communication link fault}
		The communication link faults $\vartheta^a_{ij}(t)$ and $\vartheta^b_{i}(t)$ in the directed graph, as well as their derivatives, are bounded. In addition, the signs of $a^f_{ij}(t)$ and $b^f_{i}(t)$ are the same to that of $a_{ij}$ and $b_i$. 
	\end{asm}
	
	\begin{rmk}
		Assumption \ref{6asm: communication link fault}, as also utilized in \cite{li2019robust,chen2020adaptive}, ensures the boundedness of communication faults and maintains the invariance of the network structure despite these faults. The modeling of communication link faults in (\ref{6eq: communication fault}) under Assumption \ref{6asm: communication link fault} can cover various types of communication faults and cyber attacks with bounded derivatives, such as bias attacks and fading channels.
	\end{rmk}
	
	\subsection{Formation Tracking Control Objective}
	Having modeled the MAS, taking into account input constraints and communication faults, we now proceed to formulate the formation tracking control objective. Formation refers to a specific spatial shape maintained by the followers, which is typically defined prior to executing any formation control.  To delineate a formation task, we assign each follower in the system a specific formation displacement relative to the virtual leader, denoted as $\varDelta_i$ for $i=1,2,\cdots, M$. Furthermore, let $x=\left[x_1^\top\ x_2^\top\ \cdots\ x_M^\top\right]^\top\in\mathbb{R}^{rnM}$ represent the collective state vector of all the followers. The state of the leader $0$ is extended correspondingly as ${\xi}=1_M\otimes \xi_0\in\mathbb{R}^{rnM}$. Then, a global formation tracking error can be defined
	\begin{align}\label{6eq: formation error}
		\tilde{x}=\begin{bmatrix}
			\tilde{x}_1^\top\ \tilde{x}_2^\top\ \cdots\ \tilde{x}_M^\top
		\end{bmatrix}^\top=x-{\xi}-\varDelta
	\end{align}
	where $\varDelta=\left[\varDelta_1^\top\ \varDelta_2^\top\ \cdots\ \varDelta_M^\top\right]^\top\in\mathbb{R}^{rnM}$ is the collective formation displacement vector.

	\begin{asm}
		To define a practical formation task, the displacement vector $\varDelta_i$, which encodes the desired offset between $x_i=\left[y_i^\top\ \dot{y}_i^\top\ \cdots\ y_i^{(r-1)\top}\right]^\top$ and $\xi_0$, is structured as $\left[\varDelta_{y_i}^\top\ \dot{\varDelta}_{y_i}^\top\ \cdots\ \varDelta_{y_i}^{(r-1)\top}\right]^\top$. A default assumption is that $\varDelta_{y_i}\in\mathbb{R}^n$ should be at least $(r-1)$-times differentiable, and the $r$th derivative, $\varDelta_{y_i}^{(r)}$, is considered to be zero.
	\end{asm}
	
	\begin{rmk}
		Given the previously defined directed communication topology, the knowledge of the leader's state and dynamics, as well as the desired formation displacement information, is not required to be globally known across the MAS. Only agents directly connected to the virtual leader have access to the real-time values of $\xi_0$ and the respective $\varDelta_i$. Agents that do not have a direct communication link with the virtual leader are only required to store and transmit the displacement vector relative to their out-neighbors, defined as $\varDelta_{ij}=\varDelta_i-\varDelta_j$, along with their state measurement $x_j$, to their designated out-neighbor node $j\in \mathcal{N}_i^+$ via the communication links. 
	\end{rmk}

	The control objective of this study is to develop a distributed control strategy that utilizes solely locally available neighborhood information for effective formation tracking of a MAS composed of $M$ followers (\ref{6eq: nonlinear systems}) and a virtual leader (\ref{6eq: leader dynamics}). The primary goal is to ensure that the global formation error $\tilde{x}$ not only converges to, but also remains within a small region near the origin. To ensure such formation tracking control objective is achievable, the following assumption of the graph topology holds throughout this paper.
	\begin{asm}\label{6asm: spanning tree}
		In the directed graph $\mathcal{G}$, each node is either part of a spanning tree with the root node connected to the virtual leader or a standalone node directly connected to the virtual leader.
	\end{asm}
	\begin{rmk}
		The above assumption ensures that there is either direct or indirect connectivity between each follower and the leader, providing a directed path from the leader to all the followers. This network structure is critical for achieving synchronized behaviors among the agents, enabling the distributed control strategies to eliminate the formation error across the system effectively.
	\end{rmk}
	
	\begin{lemt}{\rm~\cite{ren2011distributed}}\label{2lemt: graph 2}
		Let $\mathcal{G}$ be the directed graph for $M$ followers, labeled as agents or followers $1$ to $M$. Let $\mathcal{L}$ be the nonsymmetric Laplacian matrix associated with the directed graph $\mathcal{G}$. Suppose that in addition to the $M$ followers, there exists a leader, labeled as agent $0$, whose connection to the  $M$ followers can be described by a pinning matrix $\mathcal{B}=\text{diag}(b_1,b_2,\cdots,$ $b_M)$, where $b_i>0$ if the $i$th follower can receive information from the leader and $b_i=0$ otherwise. Let $\mathcal{L_B}=\mathcal{L}+\mathcal{B}$. Then, all eigenvalues of $\mathcal{L_B}$ have positive real parts if and only if in the directed graph $\mathcal{G}$ the leader has directed paths to all followers. 
	\end{lemt}
	
	\section{Distributed Control Design}\label{6s: Control Design}
	In this section, we present the design of an adaptive distributed MPC framework for addressing the formation tracking control problem with input constraints and communication faults.  This framework integrates state observers with MPC controllers via a distributed structure. As illustrated in Figure \ref{6fig: diagram}, each follower's local control system operates independently and relies exclusively on locally available information, consisting of an adaptive observer for estimating the leader information with resilience to communication link faults and an MPC-based controller for determining optimal formation tracking control actions online based on the local estimation of leader's state, dynamics matrix and desired displacement vector. Let us clarify and elaborate on each component in the subsequent subsections.
	\begin{figure}[t]
		\centering
		\includegraphics[width=\columnwidth]{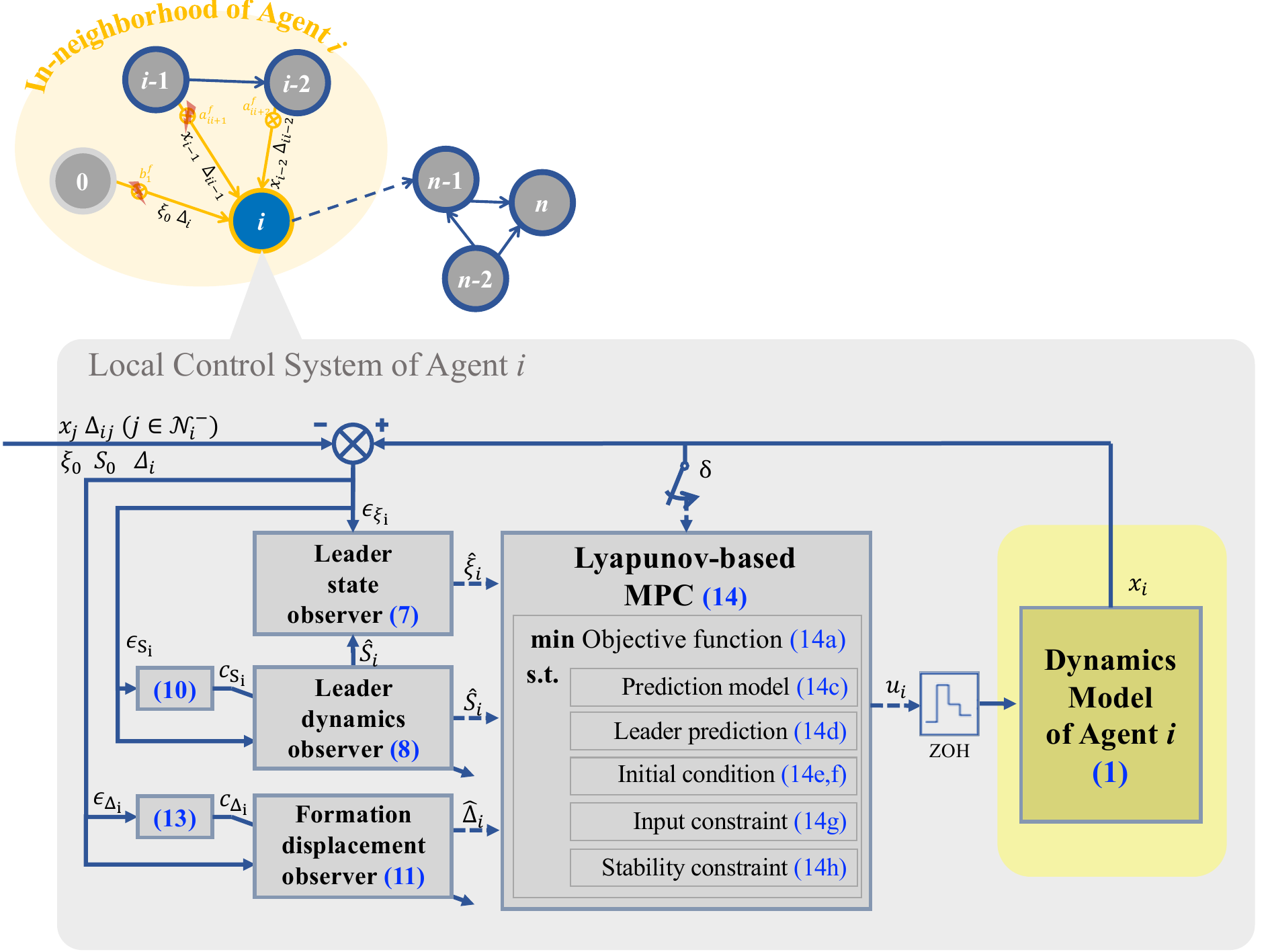}
		\caption{Detailed view of agent $i$'s local control system in the distributed network}
		\label{6fig: diagram}
	\end{figure}
	
	\subsection{Adaptive Leader Observer}
	Given the limitations on the availability of direct, real-time access to the state and dynamics of the virtual leader among all followers in the network, it becomes essential to develop an adaptive distributed observer within each local control system. This observer is responsible for estimating the leader's information and the formation displacement, which are critical components for effective formation tracking controller design.
	
	The locally estimated leader state for follower $i$ is denoted as $\hat{\xi}_i$. We can then define a  leader state estimation error as
	\begin{align}
		\epsilon_{\xi_i}&=\displaystyle\sum_{j=1}^Ma_{ij}(t)\left(\hat{\xi}_i-\hat{\xi}_j\right)+b_i(t)\left(\hat{\xi}_i-\xi_0\right)
	\end{align}
	which is available for the local control system of follower $i$. The distributed adaptive leader state observer is then designed as
	\begin{align}\label{6eq: distributed observer 1}
		\dot{\hat{\xi}}_i&=\hat{S}_i\hat{\xi}_i-c_{\xi_i}\epsilon_{\xi_i}
	\end{align}
	where $c_{\xi_i}$ is a user-designed positive observation gain. 
	
	In (\ref{6eq: distributed observer 1}), $\hat{S}_i$ is the estimate of the leader's dynamics matrix $S_0$, updated following the following estimating law
	\begin{align}\label{6eq: distributed observer 2}
		\dot{\hat{S}}_i&=-\left(c_{S_i}+\dot{c}_{S_i}\right)\epsilon_{S_i}
	\end{align}
	where ${\epsilon}_{S_i}$ is the local estimation error for $S_0$, defined as
	\begin{align}
		\epsilon_{S_i}&=\displaystyle\sum_{j=1}^Ma_{ij}(t)\left(\hat{S}_i-\hat{S}_j\right)+b_i(t)\left(\hat{S}_i-S_0\right)
	\end{align}
	In (\ref{6eq: distributed observer 2}), ${c}_{S_i}$ satisfying ${c}_{S_i}(0)\geqslant1$, is updated by
	\begin{align}
		\dot{c}_{S_i}&=\vec{\epsilon}_{S_i}^\top\vec{\epsilon}_{S_i}
	\end{align}
	with $\vec{\epsilon}_{S_i}={\rm vec}({\epsilon}_{S_i})$ being the vector form of the matrix $\epsilon_{S_i}$. The operation ${\rm vec}(\cdot)$ rearranges the matrix segments into a column vector.

	Similarly, let $\hat{\varDelta}_i$ denote the estimate of the desired formation displacement $\varDelta_i$. Its estimating law is
	\begin{align}\label{6eq: distributed observer 3}
		\dot{\hat{\varDelta}}_i&=-\left(c_{\varDelta_i}+\dot{c}_{\varDelta_i}\right)\epsilon_{\varDelta_i}
	\end{align}
	where $\epsilon_{{\varDelta}_i}$ is the local estimation errors for $\varDelta_i$, defined as
	\begin{align}
		\epsilon_{{\varDelta}_i}&=\displaystyle\sum_{j=1}^M a_{ij}(t)\left(\hat{{\varDelta}}_i-\hat{{\varDelta}}_j-{\varDelta}_{ij}\right)+b_i(t)\left(\hat{{\varDelta}}_i-{\varDelta}_{i}\right)
	\end{align}
	with ${\varDelta}_{ij}={\varDelta}_{i}-{\varDelta}_{j}$ being the desired relative displacement between follower $i$ and $j$. In (\ref{6eq: distributed observer 3}), ${c}_{\varDelta_i}$ satisfies ${c}_{\varDelta_i}\geqslant1$ and follows the following adaptive law:
	\begin{align}
		\dot{c}_{\varDelta_i}&=\epsilon_{\varDelta_i}^\top\epsilon_{\varDelta_i}
	\end{align}
	
	\subsection{MPC-based Formation Tracking Controller}
	
	With the local estimation of the leader's state $\hat{\xi}_i$, the leader's dynamics matrix $\hat{S}_i$, and the desired formation displacement vector $\hat{\varDelta}_i$, we can move on to the development of the formation tracking controller. 
	
	By applying model predictive control, a finite-horizon constrained optimization problem for a forward-looking prediction horizon $T\in\mathbb{R}^+$ is solved at each control update instant. The optimization solution is then implemented to the plant system recedingly, under a sampled-and-hold manner for execution. The control update instant sequence is $\{t_k = \delta k|k \in \{0, 1, 2, \cdots \}\}$, where $\delta<T\in\mathbb{R}^+$ represents the sampling period. 
	
	At the time instant $t_k$, the MPC optimization problem is formulated as 
	\begin{subequations}\label{6eq: MPC Problem}
		\begin{align}
			\min_{u_i^p(t|t_k)}\left( \int_{t_k}^{t_k+T}\left\|s^p_i(t|t_k)\right\|^2_{Q_i}+\left\|u_i^p(t|t_k)\right\|^2_{R_i} \right)\text{d} t
		\end{align}
		with
		\begin{align}
			\!\!\!s^p_i(t|t_k)\!=&\!\sum_{l=0}^{r-2}\lambda_{i,l}\left(x^p_{i,l+1}(t|t_k)\!-\!\xi_{0,l+1}^p(t|t_k)\!-\!\hat{\varDelta}_{i,l+1}(t_k)\right)\nonumber\\&\!+\!\left(x^p_{i,r}(t|t_k)\!-\!\xi_{0,r}^p(t|t_k)\!-\!\hat{\varDelta}_{i,r}(t_k)\right)\!\!\!
		\end{align}
		subject to
		\begin{align}
			&	\hspace{-3mm}\left\{\!\!\!
			\begin{array}{rl}
				\dot{x}_{i,1}^p(t|t_k)&\!\!\!=x_{i,2}^p(t|t_k)\\
				&\!\vdots\\
				\dot{x}_{i,r\!-\!1}^p(t|t_k)&\!\!\!=x_{i,r}^p(t|t_k)\\
				\dot{x}_{i,r}^p(t|t_k)&\!\!\!=f_i(x_i^p(t|t_k))+G_i(x_i^p(t|t_k))u_i^p(t|t_k)
			\end{array}	\right.\label{6eq: prediction model}\\
			&\hspace{6mm}\dot{\xi}_0^p(t|t_k)=\hat{S}_i(t_k){\xi}_0^p(t|t_k)\label{6eq: leader model}\\
			&\hspace{4mm}x_i^p(t_k|t_k)=x_i(t_k)
			\label{6eq: initial condition 1}\\
			&\hspace{4mm}\xi_0^p(t_k|t_k)=\xi_i(t_k)\label{6eq: initial condition 2}\\
			&\hspace{5mm}u^p_i(t|t_k)\in \Omega_{u}\label{6eq: control constraint}
		\end{align}
		\begin{align}
			&{s^p_i}^\top\!\!(t_k|t_k)\Bigl(\sum_{l=0}^{r\!-\!2}\lambda_{i,l}\left(x_{i,l+2}^p(t_k|t_k)\!-\!\xi_{0,l+2}^p(t_k|t_k)\!-\!\hat{\varDelta}_{i,l+1}(t_k)\right)\nonumber\\
			&+f_i(x_i^p(t_k|t_k)) +G_i(x_i^p(t_k|t_k))u_i^p(t_k|t_k)-\dot{\xi}^p_{0,r}(t_k|t_k)  \Bigr)\nonumber\\
			&\leqslant-c_i \|{s_i^p}(t_k|t_k)\|^2\label{6eq: Lyapunov constraint}
		\end{align}
	\end{subequations}
	where $t\in[t_k,t_k+T]$, and the internal variables are denoted by a superscript $p$ to distinguish them from the actual system signals. In the optimization problem (\ref{6eq: MPC Problem}), constraint (\ref{6eq: prediction model}) serves as the prediction model to predict the future evolution of the follower itself, while (\ref{6eq: leader model}) is to predict the leader's behavior by making use of $\hat{S}_i$ estimated by the observer. Constraints (\ref{6eq: initial condition 1}) and (\ref{6eq: initial condition 2}) specify the initial conditions of the prediction models (\ref{6eq: prediction model}) and (\ref{6eq: leader model}) , respectively. Compliance with the input constraint is ensured by (\ref{6eq: control constraint}). The Lyapunov-based stability constraint (\ref{6eq: Lyapunov constraint}) is designed to enforce the decay of the Lyapunov function at the current instant $t_k$.
	
	\begin{lemt}\label{6lemt: feasibility}
		There always exists a feasible solution to the optimization problem (\ref{6eq: MPC Problem}), constructed as
		\begin{subequations}\label{6eq: feasible solution}
			\begin{align}\label{6eq: feasible solution1}
				u_i^0(t|t_k)=& {\rm sat}\left(v_i(t|t_k),\underline{u}_i,\overline{u}_i\right)\\\label{6eq: feasible solution2}
				v_i(t|t_k)=& G_i^{-1}(x_i^p(t|t_k))\left(v_{i}^a(t|t_k)+v_{i}^d(t|t_k)\right)\\\label{6eq: feasible solution3}
				v_{i}^a(t|t_k)=& -c_is_i^p(t|t_k)\!-\!\sum_{l=0}^{r-2}\lambda_{i,l}\!\left(x_{i,l+2}^p(t|t_k)\!-\!\xi_{0,l+2}^p(t|t_k)\!\right.\nonumber\\
				&\left.-\!\hat{\varDelta}_{i,l+1}(t_k)\right)\!-\!f_i(x_i^p(t|t_k))\!+\!\dot{\xi}^p_{0,r}(t|t_k) \\\label{6eq: feasible solution4}
				v_{i}^d(t|t_k)=&\ -k_{s_i}\underline{\chi}_i^{-1}\text{sgn}\left(s_i^p(t|t_k)\right)
			\end{align}
		\end{subequations}
		for $t\in[t_k, t_k+T]$. In (\ref{6eq: feasible solution1}), ${\rm sat}\left(v_i(t),\underline{u}_i,\overline{u}_i\right)$ is a saturation function defined as
		written as
		\begin{align}\label{3eq: saturation function}
			{\rm sat}\left(v_i(t),\underline{u}_i,\overline{u}_i\right)&=\chi\left(v_i(t),\underline{u}_i,\overline{u}_i\right)v_i(t)
		\end{align}
		with
		\begin{align*}
			\chi\left(v_i(t),\underline{u}_i,\overline{u}_i\right)&\!=\!\text{diag}\left(
			\chi_j\left(v_{i_j}\!(t),\underline{u}_{i_j},\overline{u}_{i_j}\right)\right)\\
			\chi_j\left(v_{i_j}\!(t),\underline{u}_{i_j},\overline{u}_{i_j}\right)&\!=\!\left\{
			\begin{array}{cl}
				\frac{\overline{u}_{i_j}}{v_{i_j}\!(t)}, &\text{if }v_{i_j}\!(t)\geqslant \overline{u}_{i_j}\\
				1, &\text{if }-\underline{u}_{i_j}< v_{i_j}\!(t)< \overline{u}_{i_j}\\
				-\frac{\underline{u}_{i_j}}{v_{i_j}\!(t)}, &\text{if }v_{i_j}\!(t)\leqslant -\underline{u}_{i_j}
			\end{array}
			\right.
		\end{align*} 
		in which $j=1,2,\cdots,n$; $v_{c_i}$,  $\underline{u}_{i_j}$ and  $\overline{u}_{i_j}$ being the $j$th element of $v_i$, $\underline{u}_i$ and $\overline{u}_i$. $\chi_i\left(v_i,\underline{u}_i,\overline{u}_i\right)$ is the control input saturation degree indicator, and it is clear that all of its diagonal elements vary between $(0,1]$. $v_i^d$ in (\ref{6eq: feasible solution4}) is the discontinuous portion of the feasible input,  where $\underline{\chi}_i\in\mathbb{R}^{n\times n}$ is the lower bound of $\chi\left(v_i,\underline{u}_i,\overline{u}_i\right)$ satisfying
		\begin{align}\label{6eq: chi}
			O_n<\underline{\chi}_i\leqslant\chi_i\left(v_i,\underline{u}_i,\overline{u}_i\right)\leqslant I_n
		\end{align}
		The user-defined positive parameter $k_{s_i}$ is chosen appropriately such that 
		\begin{align}\label{6eq: ks}
			\left\|\left(\chi_i\left(v_i,\underline{u}_i,\overline{u}_i\right)-I_n\right)v_i^a\right\|_1\leqslant k_{s_i}
		\end{align}
	\end{lemt}
	
	\begin{proof}
		Due to the inclusion of the saturation function, the control profile constructed in (\ref{6eq: feasible solution}) naturally satisfies the input constraint (\ref{6eq: control constraint}). Substituting $u_i^0(t_k|t_k)$ into the right half of the inequality in (\ref{6eq: Lyapunov constraint}) gives
		\begin{align}
			&{s^p}_i^\top(t_k|t_k) \biggl(\sum_{i=0}^{r-2}\!\lambda_{i}\!\left(\!x^p_{i,l+2}\!(t_k|t_k)\!-\!{\xi^p_{0,l+2}}\!(t_k|t_k)\!-\!\hat{\delta}_{i,l+1}(t_k)\right)\nonumber\\
			&\!+\!f_i({x}_i^p(t_k|t_k))+G_i({x}_i^p(t_k|t_k))u_i^0(t_k|t_k)-\dot{\xi^p_{0,r}}(t_k|t_k)\biggr)\nonumber\\
			&=\!-\!c_i\|s^p_i\!(t_k|t_k)\|\!-\!k_{s_i}{s_i^p}^\top\!(t_k|t_k)\!\!\left(\!\chi_i\left(v_i(t_k),\underline{u}_i,\overline{u}_i\right)\underline{\chi}_i^{-1}\!-\!I_n\!\right)\nonumber\\
			&\text{sgn}(s_i^p(t_k|t_k))-k_{s_i}{s_i^p}^\top(t_k|t_k)\text{sgn}(s_i^p(t_k|t_k))+{s_i^p}^\top(t_k|t_k)\nonumber\\
			&\left(\chi_i\left(v_i(t_k),\underline{u}_i,\overline{u}_i\right)-I_n\right)v_i^a(t_k)
		\end{align}
		Recalling inequality (\ref{6eq: chi}), it can be obtained that $ I_n\leqslant\chi_i\left(v_i(t_k),\underline{u}_i,\overline{u}_i\right)\underline{\chi}_i^{-1}$.  Further with (\ref{6eq: ks}), we can have that
		\begin{align}
			&{s^p}_i^\top(t_k|t_k) \biggl(\sum_{i=0}^{r-2}\!\lambda_{i}\!\left(\!x^p_{i,l+2}\!(t_k|t_k)\!-\!{\xi^p_{0,l+2}}\!(t_k|t_k)\!-\!\hat{\delta}_{i,l+1}(t_k)\right)\nonumber\\
			&\!+\!f_i({x}_i^p(t_k|t_k))+G_i({x}_i^p(t_k|t_k))u_i^0(t_k|t_k)-\dot{\xi^p_{0,r}}(t_k|t_k)\biggr)\nonumber\\
			&\leqslant-c_i\|s_i^p(t_k|t_k)\|-k_{s)i}\|{s_i^p}(t_k|t_k)\|_1+\|\left(\chi_i\left(v_i(t_k),\underline{u}_i,\overline{u}_i\right)\right.\nonumber\\
			&\left.-I_n\right)v_i^a(t_k)\|_1\|{s_i^p}(t_k|t_k)\|_1\nonumber\\
			&\leqslant-c_i\|s_i^p(t_k|t_k)\|
		\end{align}
		which satisfies the Lyapunov-based stability constraint (\ref{6eq: Lyapunov constraint}). Therefore, it can be concluded that $u_i^0$ is a feasible solution to the optimization problem (\ref{6eq: MPC Problem}). 
	\end{proof}
	
	Given the feasibility of the optimization problem (\ref{6eq: MPC Problem}),  an optimal control profile $u_i^*(t|t_k)$ for $t\in[t_k,t_k+T]$ can always be found by solving it at $t_k$. The found optimal solution is then implemented in a receding horizon manner. In this regard, $u_i^*(t|t_k)$ is applied to the $i$th follower until the next measurement is available, so the actual control command $u_i(t)$ for $t\in[t_k,t_{k+1})$ is
	\begin{align}\label{6eq: actual control}
		u_i(t)=u_i^*(t_k|t_k)
	\end{align}
	When the new measurement is updated at $t_{k+1}$, the optimization problem (\ref{6eq: MPC Problem}) will be solved again with $t_k$ replaced by $t_{k+1}$, and a new optimal control profile $u_i^*(t|t_{k+1})$ for $t\in[t_{k+1}, t_{k+1}+T]$ will be found. In turn, the newly found optimal control profile updates the actual control command $u_i(t)$ for $t\in[t_{k+1}, t_{k+2})$. 
	
	Finally, we can have that with the initial conditions (\ref{6eq: initial condition 1} and \ref{6eq: initial condition 2}) specified and the stability condition (\ref{6eq: Lyapunov constraint}) satisfied, the following inequality holds for $t\in[t_k,t_{k+1})$
	\begin{align}\label{6eq: inequality}
		&{s}_i^\top\!\!(t_k) \!\biggl(\sum_{l=0}^{r-2}\lambda_{l}\left( x_{i,l+2}(t_k)\!-\!\hat{\xi}_{i,l+2}(t_k)\!\right)\!
		\!+\!f_i({x}_i(t_k))\!+\!G_i({x}(t_k))\nonumber\\&u_i(t)\!-\!\dot{\hat{\xi}}_{i,r}(t_k)\biggr)\leqslant -c_i\left\|{s}_i(t_k)\right\|^2
	\end{align}
	
	\section{Closed-loop Stability Analysis}\label{6s: Stability Analysis}
	
	Given that the local control system includes two decoupled components—an adaptive observer and an MPC-based controller—we can perform the closed-loop stability analysis in a two-step manner. This approach allows us to access the stability contributions of the observer and the controller separately. In this section, we first examine the closed-loop performance of the adaptive observer network to prove that global estimation errors converge asymptotically. Subsequently, we evaluate the MPC-based controller, verifying its ability to maintain system stability and formation tracking control performance based on the observer's estimations. 
	
	\subsection{Convergence of Estimation}
	First of all, let us focus on the estimation performance of the observers. We start by denoting collective vectors of $\hat{\xi}_i$, $\hat{S}_i$ and $\hat{\varDelta}_i$ for $i=1,2,\cdots, M$ as $\tilde{\xi}=\left[\hat{\xi}^\top_1\ \hat{\xi}^\top_2\ \cdots\ \hat{\xi}^\top_M\right]^\top$, $\hat{S}={\rm diag}(\hat{S}_1, \hat{S}_2, \cdots, \hat{S}_M)$ and $\hat{\varDelta}=\left[\hat{\varDelta}^\top_1\ \hat{\varDelta}^\top_2\ \cdots\ \hat{\varDelta}^\top_M\right]^\top$. Then, we can define the following collective estimation errors: 
	\begin{align}
		\tilde{\xi}=&\begin{bmatrix}
			\tilde{\xi}_1^\top\ \tilde{\xi}_2^\top\ \cdots\ \tilde{\xi}_M^\top
		\end{bmatrix}^\top=\hat{\xi}-\xi\\
		\tilde{S}=&{\rm diag}(\tilde{S}_1, \tilde{S}_2, \cdots, \tilde{S}_M)=\hat{S}-\left(I_M\otimes S_0\right)\\
		\tilde{\varDelta}=&\begin{bmatrix}
			\tilde{\varDelta}_1^\top\ \tilde{\varDelta}_2^\top\ \cdots\ \tilde{\varDelta}_M^\top
		\end{bmatrix}^\top=\hat{\varDelta}-\varDelta
	\end{align}
	We can also define a vector form of the leader matrix estimation error as
	\begin{align}
		\tilde{\vec{S}}\!=\!\left[\tilde{\vec{S}}_1^\top\ \tilde{\vec{S}}_2^\top\ \cdots\  \tilde{\vec{S}}_M^\top\right]^{\!\top}\!\!\!\!=&\!\left[{\rm vec}(\hat{S}_1)^{\!\top}\ {\rm vec}(\hat{S}_2)^{\!\top}\cdots {\rm vec}(\hat{S}_M)^{\!\top}\!\right]^{\!\top}\!\!\!\nonumber\\
		&-\left(1_M\otimes {\rm vec}(S_0)\right)
	\end{align}
	By recalling the adaptive distributed observer design in (\ref{6eq: distributed observer 1}), (\ref{6eq: distributed observer 2}), and (\ref{6eq: distributed observer 3}),  we can have the dynamics of these global estimation errors as follows
	\begin{align}\label{6eq: estimation error dynamics}
		\dot{\tilde{\xi}}&=\left(I_M\otimes S_0+\tilde{S}\right)\!\tilde{\xi}-\left(C_x\otimes I_{rn}\right)\!\epsilon_\xi+\tilde{S}\xi\\
		\dot{\tilde{S}}&=-\left(\left(C_S+\dot{C}_S\right)\otimes I_{rn}\right)\epsilon_S\\
		\dot{\tilde{\vec{S}}}&=-\left(\left(C_S+\dot{C}_S\right)\otimes I_{(rn)^2}\right)\vec{\epsilon}_S\\
		\dot{\tilde{\varDelta}}&=-\left(\left(C_{\!\varDelta}+\dot{C}_{\!\varDelta}\right)\otimes I_{rn}\right)\epsilon_{\!\varDelta}
	\end{align}
	where $C_x={\rm diag}(c_{\xi_1},c_{\xi_2},\cdots,c_{\xi_M})$, $C_S={\rm diag}(c_{S_1},c_{S_2},\cdots,c_{S_M})$, $C_{\!\varDelta}={\rm diag}(c_{\varDelta_1},c_{\varDelta_2},$ $\cdots,c_{\varDelta_M})$, $\epsilon_\xi=\left[\epsilon^\top_{\xi_1}\ \epsilon^\top_{\xi_2}\ \cdots\ \epsilon^\top_{\xi_M}\right]^\top$, ${\epsilon}_S={\rm diag}({\epsilon}_{S_1}, {\epsilon}_{S_2}, \cdots, {\epsilon}_{S_M})$, $\vec{\epsilon}_S=\left[\vec{\epsilon}^\top_{S_1}\ \vec{\epsilon}^\top_{S_2}\ \cdots\ \vec{\epsilon}^\top_{S_M}\right]^\top$, $\epsilon_{\!\varDelta}=\left[\epsilon^\top_{\varDelta_1}\ \epsilon^\top_{\varDelta_2}\ \cdots\ \epsilon^\top_{\varDelta_M}\right]^\top$.
	
	We introduce a new notation, $\mathcal{L^{\it f}_B}(t)=\mathcal{L}^f(t)+\mathcal{B}^f(t)$, to represent the superposition of the Laplacian matrix and the pinning matrix. This notation allows us to elucidate the relationships between global estimation errors and collective local estimation errors in the presence of communication link faults, as demonstrated below:
	\begin{align}\label{6eq: relationship 1}
		\epsilon_\xi=&\left(\mathcal{L^{\it f}_B}(t)\otimes I_{rn}\right)\tilde{\xi}\\\label{6eq: relationship 2}
		{\epsilon}_S=&\left(\mathcal{L^{\it f}_B}(t)\otimes I_{rn}\right){\tilde{S}}\\\label{6eq: relationship 3}
		\vec{\epsilon}_S=&\left(\mathcal{L^{\it f}_B}(t)\otimes I_{(rn)^2}\right){\tilde{\vec{S}}}\\\label{6eq: relationship 4}
		\epsilon_{\!\varDelta}=&\left(\mathcal{L^{\it f}_B}(t)\otimes I_{rn}\right)\tilde{\varDelta}
	\end{align}
	Subsequently, we can derive the dynamics of the local estimation errors, which are outlined below
	\begin{align}\label{6eq: epxi}
		\dot{\epsilon}_\xi&=\left(\mathcal{L^{\it f}_B}(t)\otimes I_{rn}\right)\dot{\tilde{\xi}}+\left(\dot{\mathcal{L}}_\mathcal{B}(t)\otimes I_{rn}\right)\tilde{\xi}\nonumber\\	&=\left(I_M\otimes S_0+\tilde{S}-\mathcal{L}_{\mathcal{B}}(t)C_\xi\otimes I_{rn}\right)\epsilon_\xi+\epsilon_S\xi\nonumber\\
		&\hspace{5mm}+\left(\dot{\mathcal{L}}_\mathcal{B}(t)\otimes I_{rn}\right)\tilde{\xi}\\
		\dot{\epsilon}_S&=\left(\mathcal{L}_{\mathcal{B}}(t)\otimes I_{rn}\right)\dot{\tilde{S}}+\left(\dot{\mathcal{L}_{\mathcal{B}}}(t)\otimes I_{rn}\right)\tilde{S}\nonumber\\\label{6eq: epS}
		&=-\left(\mathcal{L^{\it f}_B}(t)\left(C_S\!+\!\dot{C}_S\right)\!\otimes \!I_{rn}\right){\epsilon}_S\!+\!\left(\dot{\mathcal{L}}_\mathcal{B}(t)\!\otimes\! I_{rn}\right)\!\tilde{S}\!\!\!\\
		\dot{\vec{\epsilon}}_S&\!=\!-\!\left(\!\mathcal{L^{\it f}_B}(t)\!\left(\!C_S\!+\!\dot{C}_S\!\right)\!\otimes \!I_{(rn)^2}\right)\!\vec{\epsilon}_S\!+\!\left(\!\!\dot{\mathcal{L}}_\mathcal{B}(t)\!\otimes \!I_{(rn)^2}\!\right)\!\!\tilde{\vec{S}}\!\!\!\\\label{6eq: epDelta}
		\dot{\epsilon}_{\!\varDelta}&=\left(\mathcal{L^{\it f}_B}(t)\otimes I_{rn}\right)\dot{\tilde{\varDelta}}+\left(\dot{\mathcal{L}}_\mathcal{B}(t)\otimes I_{rn}\right)\tilde{\varDelta}\nonumber\\
		&=\!-\!\left(\mathcal{L^{\it f}_B}(t)\left(\!C_{\!\varDelta}\!+\!\dot{C}_{\!\varDelta}\!\right)\!\otimes \!I_{rn}\!\right)\!\epsilon_{\!\varDelta}\!+\!\left(\dot{\mathcal{L}}_\mathcal{B}(t)\!\otimes\! I_{rn}\!\right)\!\tilde{\varDelta}
	\end{align}
	
	Having derived the dynamics of the error, we can then move on to formulate the first theorem regarding the convergence of distributed adaptive estimation. Before proceeding, however, it is essential to establish several foundational lemmas. These lemmas are building blocks for the proof of the main theorem. 
	
	\begin{lemt}\label{6lemt: L_B}
		Under Assumptions \ref{6asm: communication link fault} and \ref{6asm: spanning tree}, the following relationships hold 
		\begin{align}
			\left\|\tilde{\xi}\right\|&\leqslant\frac{\|\epsilon_\xi\|}{\lambda_{\min}(\mathcal{L^{\it f}_B}(t))}\\
			\left\|\tilde{\vec{S}}\right\|&\leqslant\frac{\|\vec{\epsilon}_S\|}{\lambda_{\min}(\mathcal{L^{\it f}_B}(t))}\\
			\left\|\tilde{\varDelta}\right\|&\leqslant\frac{\|\epsilon_{\!\varDelta}\|}{\lambda_{\min}(\mathcal{L^{\it f}_B}(t))}
		\end{align}
	\end{lemt}
	
	\begin{proof}
		Referring back to Lemma \ref{2lemt: graph 2}, it can be obtained that the matrix $\mathcal{L^{\it f}_B}(t)$ is nonsingular and positive-definite. This guarantees that $\left(\mathcal{L^{\it f}_B}(t)\right)^{-1}$ exists, and is nonnegative. Then, it follows from (\ref{6eq: relationship 1}) - (\ref{6eq: relationship 4}) that
		\begin{align}
			\tilde{\xi}&=\left(\mathcal{L}^{f^{-1}}_\mathcal{B}(t)\otimes I_{rn}\right)\epsilon_\xi\\ \tilde{\vec{S}}&=\left(\mathcal{L}^{f^{-1}}_\mathcal{B}(t)\otimes I_{(rn)^2}\right)\vec{\epsilon}_S\\
			\tilde{\varDelta}&=\left(\mathcal{L}^{f^{-1}}_\mathcal{B}(t)\otimes I_{rn}\right)\epsilon_{\!\varDelta}
		\end{align}
		we can then have that
		\begin{align}
			\left\|\tilde{\xi}\right\|&\leqslant\left\|\left(\mathcal{L}^{f^{-1}}_\mathcal{B}(t)\otimes I_{rn}\right)\epsilon_\xi\right\|\leqslant \frac{\|\epsilon_\xi\|}{\lambda_{\min}(\mathcal{L^{\it f}_B}(t))}\\ 
			\left\|\tilde{\vec{S}}\right\|&\leqslant \left\|\left(\mathcal{L}^{f^{-1}}_\mathcal{B}(t)\otimes I_{(rn)^2}\right)\vec{\epsilon}_S\right\|\leqslant\frac{\|\vec{\epsilon}_S\|}{\lambda_{\min}(\mathcal{L^{\it f}_B}(t))}\\
			\left\|\tilde{\varDelta}\right\|&\leqslant \left\|\left(\mathcal{L}^{f^{-1}}_\mathcal{B}(t)\otimes I_{rn}\right)\epsilon_{\!\varDelta}\right\|\leqslant \frac{\|\epsilon_{\!\varDelta}\|}{\lambda_{\min}(\mathcal{L^{\it f}_B}(t))}
		\end{align}
		which prove Lemma \ref{6lemt: L_B}. 
	\end{proof}
	
	\begin{lemt}\label{6lemt: Q and P}
		Define a diagonal matrix $P(t)=\text{diag}\left(p_1(t),p_2(t),\cdots,p_M(t)\right)$ with $\left[p_1(t)\ p_2(t)\ \cdots\ p_M(t)\right]^\top=\left(\mathcal{L^{\it f}_B}(t)\right)^{-1}{1}_M$. If  Assumptions \ref{6asm: communication link fault} and \ref{6asm: spanning tree} hold, we can have that $P(t)$ is positive-definite. Furthermore, the symmetric matrix $Q(t)$ defined by
		\begin{align}
			Q(t)=P(t)\mathcal{L^{\it f}_B}(t)+\mathcal{L}_\mathcal{B}^{f\top}(t)P(t)
		\end{align}
		is positive-definite as well. Additionally, both $Q(t)$ and its time derivative are bounded. 
	\end{lemt}
	\begin{proof}
		The proof of Lemma \ref{6lemt: Q and P} can be found under Lemmas 1 and 2 in \cite{chen2020adaptive},  and is therefore omitted here for brevity.
	\end{proof}
	
	Now, we present our first main result of the closed-loop analysis by the following theorem. 
	\begin{thmt}\label{6thmt: 1}
		Suppose that Assumptions \ref{6asm: communication link fault} and \ref{6asm: spanning tree} hold. Consider the $M$-agent system with the virtual leader (\ref{6eq: leader dynamics}), interconnected via the weighted directed graph $\mathcal{G}$. Implement the leader dynamics (\ref{6eq: leader dynamics}), the distributed leader state observer (\ref{6eq: distributed observer 1}), the leader dynamics observer (\ref{6eq: distributed observer 2}) and the formation displacement observer (\ref{6eq: distributed observer 3}) for $i=1,2,\cdots, M$. If the leader state observer gain $c_{\xi_i}$ for $i=1,2,\cdots, M$ are selected such that the following condition is satisfied
		\begin{align}
			\lambda_{\min}(C_\xi)>1+\frac{\kappa^*}{\kappa_0}
		\end{align}
		where $\kappa^*=\frac{5\kappa_2}{4\kappa_0}+\frac{5\kappa_3}{\kappa_0}+\frac{5\kappa_4}{\kappa_0}+\frac{5\kappa_5}{\kappa_0}$ with $\kappa_0=\min_{\forall t\geqslant 0}\kappa_{\min}(Q(t))$, $\kappa_1\!=\!\max_{\forall t\geqslant0}\lambda_{\max}\!\left(P^2(t)\right)$, $\kappa_2\!=\!\max_{\forall t\geqslant0}\lambda_{\max}\!\left(\!\dot{P}^2(t)\right)$,  $\kappa_3=\max_{\forall t\geqslant0}\lambda_{\max}\left(P^2(t)\mathcal{L}^{f^{-\top}}_\mathcal{B}\!\!(t)\dot{\mathcal{L}}_\mathcal{B}^{f^\top}\!\!(t)\dot{\mathcal{L}}^f_\mathcal{B}(t)\mathcal{L^{\it f}_B}(t)\right)$, $\kappa_4=\max_{\forall t\geqslant0}\lambda_{\max}\left(P^2(t)\right.$ $\left.\otimes S_0^\top S_0\right)$, and $\kappa_5=\max_{\forall t\geqslant0}\lambda_{\max}\left(\tilde{S}^\top\left(P^2(t)\otimes I_{rn}\right)\tilde{S}\right)$, then all signals within the observer network are globally bounded. Moreover, all the estimated errors, $\tilde{\xi}$, $\tilde{S}$, $\tilde{\vec{S}}$ and $\tilde{\varDelta}$, asymptotically converge to the origin. 
	\end{thmt}
	\begin{proof}
		To prove the convergence of the observer network, we divide the proof into three parts. Firstly, in {Part 1} and {Part 2}, the convergence of the displacement estimation error and the dynamics matrix estimation error are proven, respectively. Finally, in {Part 3}, we can prove that the leader state estimation error asymptotically converges to the origin. 
		
		\noindent\textbf{Part 1: } To demonstrate the convergence of the formation displacement estimation error, we begin by examining its corresponding local estimation error $\epsilon_{\varDelta}$.  A Lyapunov function candidate can be selected as follows 
		\begin{align}
			V_{\!\varDelta}=&\sum_{i=1}^m\left(2c_{\varDelta_i}+\dot{c}_{\varDelta_i}\right)p_i(t)\epsilon_{\varDelta_i}^\top\epsilon_{\varDelta_i}+\sum_{i=1}^m\left( c_{\varDelta_i}-\alpha_{\!\varDelta}\right)^2\nonumber\\
			=&\epsilon_{\!\varDelta}^\top\left(\left(2C_{\!\varDelta}+\dot{C}_{\!\varDelta}\right)P(t)\otimes I_{rn}\right)\epsilon_{\!\varDelta}+\text{tr}\left(\left( C_{\varDelta}-\alpha_{\!\varDelta} I_M\right)^2\right)
		\end{align}	
		where $\alpha_{\!\varDelta}$  is a positive constant to be determined later; $P(t)$ is defined in Lemma \ref{6lemt: Q and P}. 
		
		Taking the time derivative of $V_{\!\varDelta}$ gives
		\begin{align}\label{6eq: dotVd}
			\dot{V}_{\!\varDelta}=&4\sum_{i=1}^m\left(c_{\varDelta_i}+\dot{c}_{\varDelta_i}\right)p_i(t)\epsilon_{\varDelta_i}^\top\dot{\epsilon}_{\varDelta_i}+2\sum_i^m\dot{c}_{\varDelta_i}p_i(t)\epsilon_{\varDelta_i}^\top\epsilon_{\varDelta_i}\nonumber\\
			&+\sum_{i=1}^m\left(2c_{\varDelta_i}+\dot{c}_{\varDelta_i}\right)\dot{p}_i(t)\epsilon_{\varDelta_i}^\top\epsilon_{\varDelta_i}+2\sum_{i=1}^m\left(c_{\varDelta_i}-\alpha_{\varDelta}\right)\dot{c}_{\varDelta_i}\nonumber\\
			=&4\epsilon_{\!\varDelta}^\top\left(\left(C_{\!\varDelta}\!+\!\dot{C}_{\!\varDelta}\right)\!P(t)\!\otimes \!I_{rn}\right)\dot{\epsilon}_{\!\varDelta}\!+\!2\epsilon_{\!\varDelta}^\top\left(\dot{C}_{\!\varDelta} P(t)\!\otimes\! I_{rn}\right)\!{\epsilon}_{\!\varDelta}\!\nonumber\\
			&+\!\epsilon_{\!\varDelta}^\top\left(\left(2C_{\!\varDelta}\!+\!\dot{C}_{\!\varDelta}\right)\dot{P}(t)\!\otimes\! I_{rn}\right){\epsilon}_{\!\varDelta}+2\epsilon_{\!\varDelta}^\top\left(C_{\varDelta}\otimes I_{rn}\right)\epsilon_{\!\varDelta}\nonumber\\
			&-2\alpha_{\!\varDelta}\epsilon_{\!\varDelta}^\top\epsilon_{\!\varDelta}
		\end{align}
		
		We have $\dot{c}_{\varDelta_i}\geqslant 0$ and $c_{\varDelta_i}(t)\geqslant 1$. Substituting (\ref{6eq: epDelta}) into (\ref{6eq: dotVd}) gives
		\begin{align}
			\dot{V}_{\!\varDelta}\!\!\leqslant\!&-2\kappa_0\epsilon_{\!\varDelta}^\top\!\left(\left(C_{\!\varDelta}\!+\!\dot{C}_{\!\varDelta}\right)^{\!2}\!\otimes\! I_{rn}\right)\!{\epsilon}_{\!\varDelta}\!+\!2\epsilon_{\!\varDelta}^\top\!\!\left(\!\dot{C}_{\!\varDelta} P(t)\!\otimes\! I_{rn}\!\right)\!{\epsilon}_{\!\varDelta}\nonumber\\
			&+\epsilon_{\!\varDelta}^\top\!\left(\left(2C_{\!\varDelta}+\dot{C}_{\!\varDelta}\right)\!\dot{P}(t)\otimes I_{rn}\right)\!{\epsilon}_{\!\varDelta}\!+\!2\epsilon_{\!\varDelta}^\top\!\left(C_{\varDelta}\!\otimes\! I_{rn}\right)\epsilon_{\!\varDelta}\!\nonumber\\
			&+\!4\epsilon_{\!\varDelta}^\top\!\left(\left(C_{\!\varDelta}\!+\!\dot{C}_{\!\varDelta}\right)\!\!P(t)\dot{\mathcal{L}}^f_\mathcal{B}(t)\!\otimes\! I_{rn}\right)\!\tilde{\varDelta}\!-\!2\alpha_{\!\varDelta}\epsilon_{\!\varDelta}^\top\epsilon_{\!\varDelta}\!\!\!
		\end{align}
		where $\kappa_0=\min_{\forall t\geqslant 0}\kappa_{\min}(Q(t))$. 
		
		Applying Young's inequality, one has
		\begin{subequations}
			\begin{align}
				&\!\!\!\!2\epsilon_{\!\varDelta}^\top\!\left(\dot{C}_{\!\varDelta}{P}(t)\!\otimes\! I_{rn}\right)\!{\epsilon}_{\!\varDelta}\!\leqslant \frac{\kappa_0}{2}\epsilon_{\!\varDelta}^\top\!\left(\dot{C}_{\!\varDelta}^2\!\otimes\! I_{rn}\right)\!{\epsilon}_{\!\varDelta}\!\nonumber\\
				&\hspace{35mm}+\!\frac{2}{\kappa_0}\epsilon_{\!\varDelta}^\top\!\left({P}^2(t)\!\otimes \!I_{rn}\right)\!{\epsilon}_{\!\varDelta}\\
				&\!\!\!\!\epsilon_{\!\varDelta}^\top\!\!\left(\left(2C_{\!\varDelta}\!+\!\dot{C}_{\!\varDelta}\right)\!\dot{P}(t)\!\otimes\! I_{rn}\right)\!\!{\epsilon}_{\!\varDelta}\!\leqslant \frac{\!\kappa_0\!}{4}\epsilon_{\!\varDelta}^\top\!\!\left(\left({C}_{\!\varDelta}^2\!+\!\dot{C}_{\!\varDelta}^2\right)\!\!\otimes\! I_{rn}\right)\!{\epsilon}_{\!\varDelta}\!\nonumber\\
				&\hspace{45mm}+\!\frac{5}{\!\kappa_0\!}\epsilon_{\!\varDelta}^\top\!\!\left(\!
				\dot{P}^2(t)\!\otimes\! I_{rn}\right)\!{\epsilon}_{\!\varDelta}\!\!\!\!\!\\
				&\!\!\!\!2\epsilon_{\!\varDelta}^\top\left(C_{\varDelta}\otimes I_{rn}\right)\epsilon_{\!\varDelta}\leqslant \frac{\kappa_0}{2}\epsilon_{\!\varDelta}^\top\left({C}_{\!\varDelta}^2\otimes I_{rn}\right){\epsilon}_{\!\varDelta}+\frac{2}{\kappa_0}\epsilon_{\!\varDelta}^\top{\epsilon}_{\!\varDelta}\\
				&\!\!\!\!4\epsilon_{\!\varDelta}^\top\!\!\left(\!\left(\!C_{\!\varDelta}\!+\!\dot{C}_{\!\varDelta}\!\right)\!\!P\!(t)\dot{\mathcal{L}}^f_{\mathcal{B}}(t)\!\otimes\!I_{rn}\right)\!\tilde{\varDelta}\!\leqslant \!\!\frac{\kappa_0}{4}\epsilon_{\!\varDelta}^\top\!\!\Bigl(\!\left(\!{C}_{\!\varDelta}\!+\!\dot{C}_{\!\varDelta}\!\right)^2\!\!\!\otimes\! I_{rn}\!\Bigr)\!{\epsilon}_{\!\varDelta}\nonumber\\
				&\!\!\!\!\!+\!\frac{16}{\kappa_0}\epsilon_{\!\varDelta}^\top\!\!\left(\!P^2(t)\mathcal{L}^{f^{-\top}}_\mathcal{B}\!\!(t)\dot{\mathcal{L}}_\mathcal{B}^{f^\top}\!\!(t)\dot{\mathcal{L}}^f_\mathcal{B}(t)\mathcal{L^{\it f}_B}(t)\!\otimes\! I_{rn}\!\right)\!\epsilon_{\!\varDelta}\!\!\!
			\end{align}
		\end{subequations}
		
		We define $\kappa_1=\max_{\forall t\geqslant0}\lambda_{\max}\left(P^2(t)\right)$, $\kappa_2=\max_{\forall t\geqslant0}\lambda_{\max}\left(\dot{P}^2(t)\right)$, and $\kappa_3=\max_{\forall t\geqslant0}\lambda_{\max}\left(P^2(t)\mathcal{L}^{f^{-\top}}_\mathcal{B}\!\!(t)\dot{\mathcal{L}}_\mathcal{B}^{f^\top}\!\!(t)\dot{\mathcal{L}}^f_\mathcal{B}(t)\mathcal{L^{\it f}_B}(t)\right)$. Then, substituting the above inequalities into $\dot{V}_{\!\varDelta}$ yields
		\begin{align}
			\dot{V}_{\!\varDelta}\!\leqslant\! -\!{\lambda}_0\epsilon_{\!\varDelta}^\top\!\Bigl(\!\left(\!C_{\!\varDelta}\!+\!\dot{C}_{\!\varDelta}\!\right)^{\!2}\!\!\otimes\! I_{rn}\Bigr){\epsilon}_{\!\varDelta}+{\kappa}{\epsilon}_{\!\varDelta}^\top{\epsilon}_{\!\varDelta}-2\alpha_{\!\varDelta}{\epsilon}_{\!\varDelta}^\top{\epsilon}_{\!\varDelta}
		\end{align}
		where $\kappa=\frac{2\kappa_1}{\kappa_0}+\frac{5\kappa_2}{\kappa_0}+\frac{2}{\kappa_0}+\frac{16\kappa_3}{\kappa_0}$. There exists a bounded constant $\alpha_{\!\varDelta}$ satisfying $\alpha_{\!\varDelta}\geqslant\frac{\kappa}{2}$ such that
		\begin{align}\label{6eq: dotV}
			\dot{V}_{\!\varDelta}\leqslant -\epsilon_{\!\varDelta}^\top\!\left(\left(C_{\!\varDelta}+\dot{C}_{\!\varDelta}\right)^{\!2}\!\otimes\! I_{rn}\right)\!{\epsilon}_{\!\varDelta}
		\end{align}
		which implies that all signals in the developed observer including $\epsilon_{\!\varDelta}$, $C_{\!\varDelta}$ and $\dot{C}_{\!\varDelta}$ are globally bounded under communication link faults.
		
		To get the convergence of $\epsilon_{\!\varDelta}$, we solve (\ref{6eq: dotV}) as
		\begin{align}
			V_{\!\varDelta}(t)\!-\!V_{\!\varDelta}(0)&\!\leqslant\!-{\lambda}_0\!\int_{0}^{t}\!\epsilon_{\!\varDelta}^\top\!\left(\left(C_{\!\varDelta}\!+\!\dot{C}_{\!\varDelta}\right)^{\!2}\!\otimes\! I_{rn}\!\right)\!{\epsilon}_{\!\varDelta} \text{d}\tau
		\end{align}
		
		With the fact that $\dot{c}_{\varDelta_i}\geqslant0$ and ${c}_{\varDelta_i}\geqslant1$, one has
		\begin{align}
			0\leqslant V_{\!\varDelta}(t)&\leqslant -{\lambda}_0\int_{0}^{t}\epsilon_{\!\varDelta}^\top\left(C_{\!\varDelta}\otimes I_{rn}\right){\epsilon}_{\!\varDelta} \text{d}\tau+V_{\!\varDelta}(0)\nonumber\\
			&\leqslant-\kappa_0\int_{0}^{t}\epsilon_{\varDelta}^\top\epsilon_{\varDelta}\text{d}\tau+V_{\!\varDelta}(0)
		\end{align}
		or equivalently 
		\begin{align}
			0\leqslant\int_{0}^{t}\epsilon_{\varDelta}^\top\epsilon_{\varDelta}\text{d}\tau\leqslant \frac{V_{\!\varDelta}(0)}{\kappa_0}
		\end{align}
		This demonstrates that $\int_{0}^{t}\epsilon_{\varDelta}^\top\epsilon_{\varDelta}\text{d}\tau$ is bounded. By applying Barbalat's Lemma, we establish that $\epsilon_{\varDelta}$ converges to zero. Furthermore, according to Lemma \ref{6lemt: L_B}, it can be deduced that the estimation error $\tilde{\varDelta}$ also globally converges to zero, even in scenarios involving communication link faults. 
		
		\noindent\textbf{Part 2: } Similarly, the leader state estimation error convergence can be proven. We firstly select a candidate Lyapunov function as
		\begin{align}
			\!\!\!V_S\!=\!&\sum_{i=1}^m\left(2c_{S_i}+\dot{c}_{S_i}\right)p_i(t)\vec{\epsilon}_{S_i}^\top\vec{\epsilon}_{S_i}+\sum_{i=1}^m\left( c_{S_i}-\alpha_S\right)^2\nonumber\\
			\!=&\vec{\epsilon}_S^\top\!\!\left(\left(2C_S\!+\!\dot{C}_S\right)\!P(t)\!\otimes\! I_{(rn)^2}\!\right)\!\vec{\epsilon}_S\!+\!\text{tr}\!\left(\!\left(C_S\!-\!\alpha_S I_{(rn)^2}\right)^2\right)\!\!
		\end{align}	
		
		Following the same procedure, the derivative of $V_S$ can be bounded by 
		\begin{align}
			\dot{V}_S\leqslant -{\lambda}_0\vec{\epsilon}_S^\top\!\left(\left(C_S+\dot{C}_S\right)^{\!2}\!\otimes\! I_{(rn)^2}\right)\!\vec{\epsilon}_S
		\end{align}
		with $\alpha_S$ appropriately chosen such that $\alpha_S\geqslant\frac{\kappa}{2}$.
		
		We can further have that
		\begin{align}
			0\leqslant\int_{0}^{t}\vec{\epsilon}_{S}^\top\vec{\epsilon}_{S}\text{d}\tau\leqslant \frac{V_S(0)}{\kappa_0}
		\end{align}
		This implies that $\int_{0}^{t}\vec{\epsilon}_S^\top\vec{\epsilon}_S\text{d}\tau$ is bounded. Applying Barbalat's Lemma, we can conclude that $\vec{\epsilon}_S$ converges to zero. Furthermore, based on Lemma \ref{6lemt: L_B}, it can be concluded that the estimation errors ${\epsilon}_S$, $\tilde{\vec{S}}$, and $\tilde{S}$ also converge to zero. 
		
		\noindent\textbf{Part 3: }Finally, to demonstrate the stability and convergence of the leader state estimation, we select the following candidate for a Lyapunov function:
		\begin{align}
			V_\xi=&\sum_{i=1}^mp_i(t)\epsilon_{\xi_i}^\top\epsilon_{\xi_i}=\epsilon_\xi^\top\left(P(t)\otimes I_{rn}\right)\epsilon_\xi
		\end{align}	
		
		By taking the time derivative of $V_\xi$, we can have
		\begin{align}
			\dot{V}_\xi=&2\epsilon_\xi^\top\left(P(t)\otimes I_{rn}\right)\dot{\epsilon}_\xi+\epsilon_\xi^\top\left(\dot{P}(t)\otimes I_{rn}\right){\epsilon}_\xi
		\end{align}
		Substituting (\ref{6eq: epxi}) into it gives 
		\begin{align}
			\!\!\!\!\dot{V}_\xi\!\leqslant\!&-\kappa_0\lambda_{\min}(C_\xi)\epsilon_\xi^\top\!{\epsilon}_\xi+\epsilon_\xi^\top\!\left(\dot{P}(t)\otimes I_{rn}\right)\!{\epsilon}_\xi+2\epsilon_\xi^\top\!\left(P(t)\dot{\mathcal{L}}_\mathcal{B}(t)\right.\nonumber\\
			&\left.\otimes I_{rn}\right)\!\tilde{\xi}+2\epsilon_x^\top\!\left(P(t)\otimes S_0\right)\epsilon_\xi+2\epsilon_\xi^\top\left(P(t)\otimes I_{rn}\right)\tilde{S}{\epsilon}_\xi\nonumber\\
			&+2\epsilon_\xi^\top\left(P(t)\otimes I_{rn}\right)\epsilon_S\xi
		\end{align}
		where $\kappa_0=\min_{\forall t\geqslant 0}\lambda_{\min}(Q(t))$. 
		
		We define $\varsigma=\epsilon_S \xi$ as a new variable that converges to zero given the convergence of $\epsilon_S$. Applying Young's inequality, one has
		\begin{subequations}
			\begin{align}
				\!\!\!\!\epsilon_\xi^\top\!\left(\dot{P}(t)\!\otimes\!I_{rn}\right)\!{\epsilon}_\xi\!\leqslant\!&\frac{\kappa_0}{5}\epsilon_\xi^\top{\epsilon}_\xi\!+\!\frac{5}{4\kappa_0}\epsilon_\xi^\top\!\!\left(\!\dot{P}^2(t)\!\otimes \!I_{rn}\!\right)\!\epsilon_\xi\!\!\!\!\\
				\!\!\!\!2\epsilon_\xi^\top\!\!\!\left(\!P(t)\dot{\mathcal{L}}^f_{\mathcal{B}}(t)\!\otimes\! I_{rn}\!\right)\!\tilde{\xi}\!\leqslant& \frac{\kappa_0}{5}\epsilon_\xi^\top{\epsilon}_\xi\!+\!\frac{5}{\kappa_0}\epsilon_{\xi}^\top\!\!\left(\!P^2(t)\mathcal{L}^{f^{-\top}}_\mathcal{B}\!\!(t)\dot{\mathcal{L}}_\mathcal{B}^{f^\top}\!\!(t)\right.\nonumber\\
				&\left.\dot{\mathcal{L}}^f_\mathcal{B}(t)\mathcal{L^{\it f}_B}\!(t)\!\otimes\! I_{rn}\right)\!\epsilon_{\xi}\\
				\!\!\!\!2\epsilon_\xi^\top\!\left(P(t)\otimes S_0\right)\!\epsilon_\xi\leqslant& \!\frac{\kappa_0}{5}\epsilon_\xi^\top\!{\epsilon}_\xi\!+\!\frac{5}{\kappa_0}\epsilon_\xi^\top\!\left(\!P^2(t)\!\otimes\! S_0^\top \!S_0\!\right)\!\epsilon_\xi\!\!\!\\
				\!\!\!\!2\epsilon_\xi^\top\!\left(P(t)\otimes I_{rn}\right)\!\tilde{S}\epsilon_\xi\leqslant& \!\frac{\kappa_0}{5}\!\epsilon_\xi^\top\!{\epsilon}_\xi\!+\!\frac{5}{\kappa_0}\!\epsilon_\xi^\top\!\tilde{S}^\top\!\!\!\left(\!P^2\!(t)\!\otimes \!I_{rn}\!\right)\!\!\tilde{S}\!\epsilon_\xi\!\!\!\\
				\!\!\!\!2\epsilon_\xi^\top\!\left( P(t)\otimes I_{rn}\right)\!\varsigma\leqslant&\!\frac{\kappa_0}{5}\epsilon_\xi^\top{\epsilon}_\xi\!+\!\frac{5}{\kappa_0}\varsigma^\top\!\!\left(P^2(t)\!\otimes \!I_{rn}\right)\!\varsigma\!\!
			\end{align}
		\end{subequations}
		
		Given the boundedness of $\tilde{S}$, we can define new constants $\kappa_4=\max_{\forall t\geqslant0}\lambda_{\max}\left(P^2(t)\right.$ $\left.\otimes S_0^\top S_0\right)$ and $\kappa_5=\max_{\forall t\geqslant0}\lambda_{\max}\left(\tilde{S}^\top\left(P^2(t)\otimes I_{rn}\right)\tilde{S}\right)$, and recall the definitions of 
		$\kappa_1=\max_{\forall t\geqslant0}\lambda_{\max}\left(P^2(t)\right)$, $\kappa_2=\max_{\forall t\geqslant0}\lambda_{\max}\left(\dot{P}^2(t)\right)$,  $\kappa_3=\max_{\forall t\geqslant0}\lambda_{\max}\left(P^2(t)\mathcal{L}^{f^{-\top}}_\mathcal{B}\!\!(t)\right.$ $\left.\dot{\mathcal{L}}_\mathcal{B}^{f^\top}\!\!(t)\dot{\mathcal{L}}^f_\mathcal{B}(t)\mathcal{L^{\it f}_B}(t)\right)$.  Substituting the above inequalities into $\dot{V}_{\!\xi}$ yields
		\begin{align}\label{6eq: stability condition}
			\dot{V}_\xi\leqslant& -\left(\kappa_0\lambda_{\min}(C_\xi)-\kappa_0-\kappa^*\right)\epsilon_\xi^\top\!{\epsilon}_\xi+\frac{5\kappa_1}{\kappa_0}\varsigma^\top\varsigma
		\end{align}
		where $\kappa^*=\frac{5\kappa_2}{4\kappa_0}+\frac{5\kappa_3}{\kappa_0}+\frac{5\kappa_4}{\kappa_0}+\frac{5\kappa_5}{\kappa_0}$.
		
		Therefore, when the user-designated gains $c_{\xi_i}$ for $i=1,2,\cdots, M$ are appropriately chosen to satisfy the stability condition (\ref{6eq: stability condition}), a positive constant $\alpha_\xi$ must exist such that $0<\alpha_\xi\leqslant\kappa_0\lambda_{\min}(C_\xi)-\kappa_0-\kappa^*$. As a result, one has
		\begin{align}
			\dot{V}_\xi\leqslant& -\alpha_\xi\epsilon_\xi^\top\!{\epsilon}_\xi+\frac{5\kappa_1}{\kappa_0}\varsigma^\top\varsigma
		\end{align}
		which means that the dynamics of $\epsilon_\xi$ is robust input-to-state stable with $\varsigma$ as a disturbances input. Integrating the above inequalities over $[0,t]$ gives
		\begin{align}
			V_\xi(t)-V_\xi(0)\leqslant-\alpha_\xi\int_0^t \epsilon_\xi^\top \epsilon_\xi \text{d}\tau+\frac{5\kappa_1}{\kappa_0}\int_0^t\varsigma^\top\varsigma \text{d}\tau
		\end{align}  
		which further yields
		\begin{align}
			\alpha_\xi	\int_0^t \epsilon_\xi^\top \epsilon_\xi \text{d}\tau\leqslant\frac{5\kappa_1}{\kappa_0}\int_0^t\varsigma^\top\varsigma \text{d}\tau
		\end{align}
		From {Part 2}, we have proven the convergence of $\epsilon_S$,  which also implies the convergence of $\varsigma$ to zero. Hence, it can be concluded that   $\int_0^t\varsigma^\top\varsigma \text{d}\tau$ is bounded. From the above inequality, we can prove the boundedness of $\int_0^t\epsilon_\xi^\top\epsilon_\xi \text{d}\tau$. By using Barbalat's Lemma, $\epsilon_\xi$ can be proven to converge to zero as well. From Lemma \ref{6lemt: L_B}, it can be obtained that $\tilde{\xi}$ also asymptotically converges to the origin. 
	\end{proof}
	
	\subsection{Stability of Control}
	Following the establishment of the asymptotic convergence of the estimation errors as demonstrated in Theorem \ref{6thmt: 1}, we can now present our second main result of this paper, which summarizes the convergence of the system's actual state to the locally estimated leader state under the proposed Lyapunov-based MPC framework. 
	
	We start by defining a sliding mode tracking control error for the follower $i$ as
	\begin{align}
		\!\!s_i\!=\!\sum_{l=0}^{r-2}\!\lambda_{i,l}\!\left(x_{i,l+1}\!-\!\hat{\xi}_{i,l+1}\!-\!\hat{\Delta}_{i,l+1}\!\right)\!+\!\left(x_{i,r}\!-\!\hat{\xi}_{i,r}\!-\!\hat{\Delta}_{i,r}\!\right)\!\!\!
	\end{align}
	Then, the following Lyapunov function candidate for follower $i$ can be considered
	\begin{align}
		V_i=\frac{1}{2}\|s_i\|^2
	\end{align}
	Differentiating $V$ along the closed-loop dynamics yields 
	\begin{align}\label{6eq: dotV1}
		\dot{V}_i=s_i^\top \biggl(&\sum_{l=0}^{r-2}\lambda_{l}\left( x_{i,l+2}-\hat{\xi}_{i,l+2}-\hat{\Delta}_{i,l+2}\right)
		+f_i({x}_i)\nonumber\\
		&+G_i({x}_i)u_i-\dot{\hat{\xi}}_{i,r}\biggr)
	\end{align}
	Recalling the inequality (\ref{6eq: inequality}), adding and subtracting the right-hand side of it to and from (\ref{6eq: dotV1}), we can rewrite $\dot{V}_i(t)$ over the interval $t\in[t_k,t_{k+1})$ as follows:
	\begin{align}\label{6eq: dotV2}
		\!\!\dot{V}_i(t)\!\leqslant&\!-\!c_i\|s_i(t)\|^2 \!+\!c_i\left(\|s_i(t)\|^2\!-\!\|s_i(t_k)\|^2\right)\!
		+\!s_i^\top\!(t) \biggl(\sum_{l=0}^{r-2}\lambda_{l}\nonumber\\
		&\!\!\!\left( \!x_{i,l+2}(t)\!-\!\hat{\xi}_{i,l+2}(t)\!-\!\hat{\Delta}_{i,l+2}(t)\right)\!+\!f_i({x}_i(t))\!+\!G_i({x}_i(t))\nonumber\\
		&\!\!u_i(t)\!-\!\dot{\hat{\xi}}_{i,r}(t)\biggr)\!-\!s_i^\top(t_k)\!\biggl(\sum_{l=0}^{r-2}\lambda_{l}\left( \!x_{i,l+2}(t_k)\!-\!\hat{\xi}_{i,l+2}(t_k)\!\right.\nonumber\\
		&\left.\!\!\!\!\!-\!\hat{\Delta}_{i,l\!+\!2}(t_k)\!\right)\!\!+\!\!f_i({x}_i(t_k))\!+\!G_{\!i}\!({x}(t_k))u_i(t)\!-\!\dot{\hat{\xi}}_{i,r}(t_k)\!\biggr)\!\!\!\!\!\!\!\!
	\end{align}
	
	By invoking Lipschitz continuity and under Assumption \ref{6asm: Lipschitiz continuity}, there must exist positive Lipschitz constants $L_{s_i^2}$, $L_{x_i}$ and $L_{\xi_i}$ such that 
	\begin{subequations}\label{6eq: LM}
		\begin{align}
			&\left|\|s_i(t)\|^2\!-\!\|s_i(t_k)\|^2\right|\leqslant L_{s_i^2}\left\|x_i(t)-x_i(t_k)\right\|\\
			&\biggl|s_i^\top\!\!(t)\! \biggl(\!\sum_{l\!=\!0}^{r\!-\!2}\!\lambda_{l}\!\left(\! x_{i,l\!+\!2}(t)\!-\!\hat{\xi}_{i,l\!+\!2}(t)\!-\!\hat{\Delta}_{i,l\!+\!2}(t)\!\right)\!+\!f_i({x}_i(t))\!+\!G_i({x}_i(t))\nonumber\\
			&u_i(t)\!-\!\dot{\hat{\xi}}_{i,r}(t)\!\biggr) \!\!-\!\!s_i^{\!\top\!}(t_k) \!\biggl(\sum_{l=0}^{r-2}\lambda_{l}\!\left( x_{i,l+2}(t_k)\!-\!\hat{\xi}_{i,l+2}(t_k)\!-\!\hat{\Delta}_{i,l+2}(t_k)\!\right)\!
			\nonumber\\
			&+\!\!f_i({x}_i(t_k))\!+\!G_i({x}(t_k))u_i(t)\!-\!\dot{\hat{\xi}}_{i,r}(t_k)\!\biggr)\!\biggr|\nonumber\\
			&\leqslant L_{x_i}\left\|x_i(t)-x_i(t_k)\right\|+L_{\xi_i}\left\|\hat{\xi}_i(t)-\hat{\xi}_i(t_k)\right\|
		\end{align}
		We also have positive constants $M_{x_i}$ and $M_{\xi_i}$ satisfying
		\begin{align}
			\left\|x_i(t)-x_i(t_k)\right\|&\leqslant M_{x_i}\delta\\
			\left\|\hat{\xi}_i(t)-\hat{\xi}_i(t_k)\right\|&\leqslant M_{\xi_i}\delta
		\end{align}
	\end{subequations}
	By substituting (\ref{6eq: LM}) into (\ref{6eq: dotV2}), we have that
	\begin{align}\label{6eq: dotV3}
		\dot{V}_i(t)\leqslant&-c_i\|s_i(t)\|^2+(\kappa_{i_1}+c_i\kappa_{i_2})\delta
	\end{align}
	where $\kappa_{i_1}=L_{x_i}M_{x_i}+L_{\xi_i}M_{\xi_i}$ and $\kappa_{i_2}=L_{s_i^2}M_{x_i}$.

	We define a collective sliding mode tracking control error as $s\!=\!\left[s_1^\top\ s_2^\top\ \cdots\ s_M^\top\right]^{\!\top}\!\!$, and Theorem \ref{6thmt: 2} below encapsulates the second main result of this paper 
	\begin{thmt}\label{6thmt: 2}
		Suppose Assumptions \ref{6asm: Lipschitiz continuity}-\ref{6asm: spanning tree} hold. Consider the MASs with $M$ followers  (\ref{6eq: nonlinear systems}) and a virtual leader \ref{6eq: leader dynamics} in closed loop under the developed adaptive distributed observer-based Lyapunov-based MPC framework with the leader observer (\ref{6eq: distributed observer 1})-(\ref{6eq: distributed observer 3}) and the MPC problem (\ref{6eq: MPC Problem}). 
		If $s(0)\in\Omega_{\rho_s^0}\triangleq\left\{s|V\leqslant\rho_s^0\right\}$ and the following stability condition is satisfied by choosing appropriate control gains $c_i$ for $i=1,2,\cdots, M$, 
		\begin{align}\label{6eq: stability condition2}
			&-2\rho_s\lambda_{\min}(C)+\left(\kappa_1+\lambda_{\max}(C)\kappa_2\right)\delta\leqslant0
		\end{align}
		where $\rho_s\leqslant\rho_s^0$, $C\!=\!\text{diag}(c_1,c_2,\cdots, c_M)$, $\kappa_{1}\!=\!\max(\kappa_{1_1}, \kappa_{2_1},\cdots, \kappa_{M_1})$ and $\kappa_{2}\!=\!\max(\kappa_{1_2}, \kappa_{2_2},\cdots,\kappa_{M_2})$, then, the sliding mode error $s$ of the closed-loop system is always bounded and ultimately converges to  $\Omega_{\rho_s}\triangleq\left\{s|V\leqslant\rho_s\right\}$. 
	\end{thmt}
	\begin{proof}
		A global Lyapunov function for the entire MAS can be selected as
		\begin{align}
			V=\frac{1}{2}\|s\|^2=\sum_{i=1}^M\frac{1}{2}\|s_i\|^2
		\end{align}
		Recalling (\ref{6eq: dotV3}), we can obtain that
		\begin{align}\label{6eq: dotV4}
			\dot{V}\leqslant&-\lambda_{\min}(C)\|s(t)\|^2+(\kappa_{1}+\lambda_{\max}(C)\kappa_{2})\delta
		\end{align}
		where $C\!=\!\text{diag}(c_1,c_2,\cdots, c_M)$, $\kappa_{1}\!=\!\max(\kappa_{1_1}, \kappa_{2_1},\cdots, \kappa_{M_1})$ and $\kappa_{2}\!=\!\max(\kappa_{1_2}, \kappa_{2_2},\cdots,$ $ \kappa_{M_2})$. 
		
		From the definition of $V$, we further have
		\begin{align}
			\dot{V}\leqslant-\lambda_{\min}(C)V+(\kappa_{1}+\lambda_{\max}(C)\kappa_{2})\delta
		\end{align}
		If the condition (\ref{6eq: stability condition2}) is satisfied, then it can be derived that $	\dot{V}<0$ for all $s\in\left\{s|\rho_s^0<V\leqslant\rho_{s}\right\}$ and $\dot{V}=0$ for $V=\rho_s$. Therefore, it implies that $s$ converges to $\Omega_{\rho_s}$ without leaving the stability region $\Omega_{\rho^0_s}$ as $t$ approaches $\infty$ 
	\end{proof}

	Based on Theorem \ref{6thmt: 2}, we can conclude that the state variable ${x}$ converges towards $\hat{\xi}+\hat{\Delta}$. Given the previously validated asymptotic convergence of $\hat{\xi}$ to $\xi$ and $\hat{\Delta}$ to $\Delta$ in Theorem \ref{6thmt: 1}, we can conclusively demonstrate the ultimate boundedness and convergence of the global formation tracking error $\tilde{x}$, as defined in (\ref{6eq: formation error}). This conclusion is drawn by combining the results from both Theorem \ref{6thmt: 1} and Theorem \ref{6thmt: 2}.
	\begin{rmk}
		It is crucial to recognize that the analytical convergence error arises from the sampled-and-hold implementation of the employed MPC. The asymptotically stable nature of the distributed observer network does not compromise the ultimate accuracy of formation tracking by providing sufficiently accurate estimations to the controller. 
	\end{rmk}
	\section{Simulation Study}\label{6s: Simulation Study}
	Simulation studies on two different examples are carried out to evaluate the performance of the proposed method in achieving formation tracking control under input constraints and communication link faults. This section provides the simulation results and 
	
	\subsection{Example 1: A  Numerical Multi-agent System}
	We first consider a numerical example---a nonlinear MAS with 3 followers and a leader node $0$. The $3$ followers can be described by third-order nonlinear systems as follows
	\begin{subequations}
		\begin{align}
			&	\left\{
			\begin{array}{rl}
				\dot{x}_{1,1}&\!\!=x_{1,2}\\
				\dot{x}_{1,2}&\!\!=x_{1,3}\\
				\dot{x}_{1,3}&\!\!=x_{1,1}x_{1,2}+x_{1,3}-x_{1,1}^3+u_1\\
				y_1&\!\!=x_{1,1}
			\end{array}
			\right.\\
			&	\left\{
			\begin{array}{rl}
				\dot{x}_{2,1}&\!\!=x_{2,2}\\
				\dot{x}_{2,2}&\!\!=x_{2,3}\\
				\dot{x}_{2,3}&\!\!= x_{2,1}\sin(x_{2,2})+\cos^2(x_{2,3}) + u_2\\
				y_2&\!\!=x_{2,1}
			\end{array}
			\right.	\\
			&	\left\{
			\begin{array}{rl}
				\dot{x}_{3,1}&\!\!=x_{2,2}\\
				\dot{x}_{3,2}&\!\!=x_{2,3}\\
				\dot{x}_{3,3}&\!\!= -\frac{1}{2}(x_{3,1}+x_{3,2}-1)^2(x_{3,3}-1) + u_3\\
				y_3&\!\!=x_{3,1}
			\end{array}
			\right.
		\end{align}
	\end{subequations}
	with the input constraint defined as
	\begin{align}\label{6eq: EX1 input constraint}
		\Omega_{u_{i}}=\left\{u_{i}|-3\leqslant u_{i} \leqslant 3 \right\}
	\end{align}  
	for $i=1,2,3$.
	The initial conditions of the 3 followers are $x_1(0)=[1.3\ 0\ 0]^\top$, $x_2(0)=[0.5\ 0\ 0]^\top$, $x_3(0)=[0\ 0\ 0]^\top$.  
	
	\begin{figure}[t]
		\centering
		\includegraphics[width=\columnwidth]{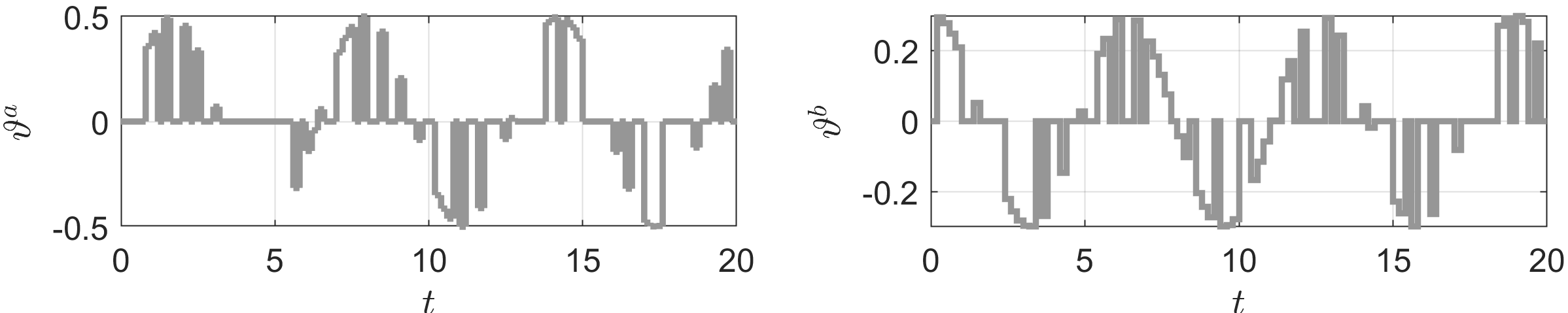}
		\caption{Communication link faults}
		\label{6fig: EX1_th}
	\end{figure}
	\begin{figure}[t]
		\centering
		\includegraphics[width=\columnwidth]{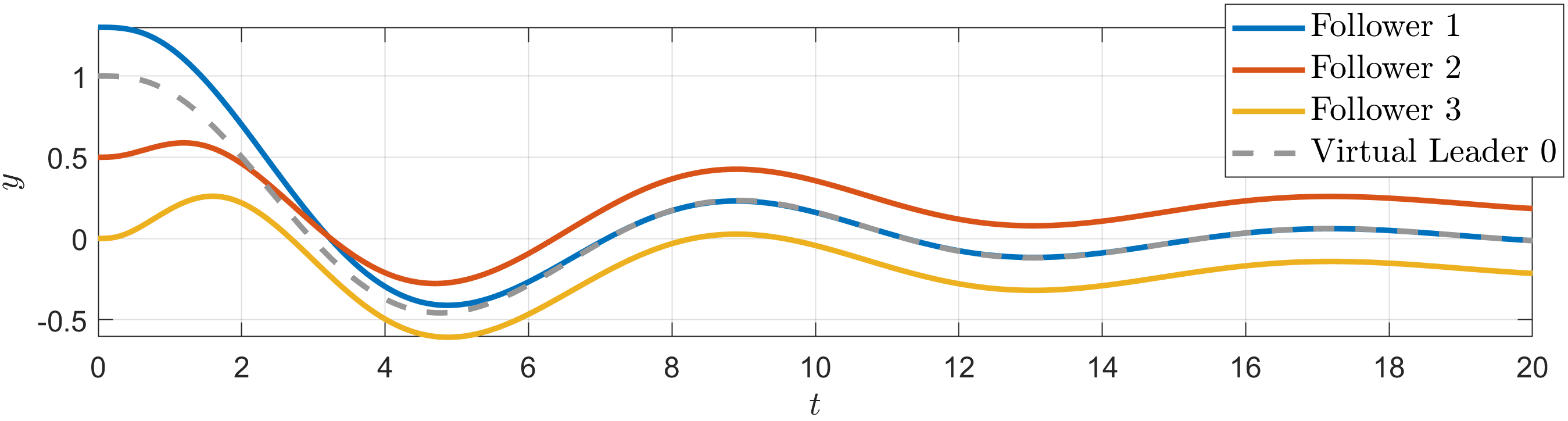}
		\caption{Formation tracking performance}
		\label{6fig: EX1_y}
	\end{figure}
	\begin{figure}[t]
		\centering
		\includegraphics[width=\columnwidth]{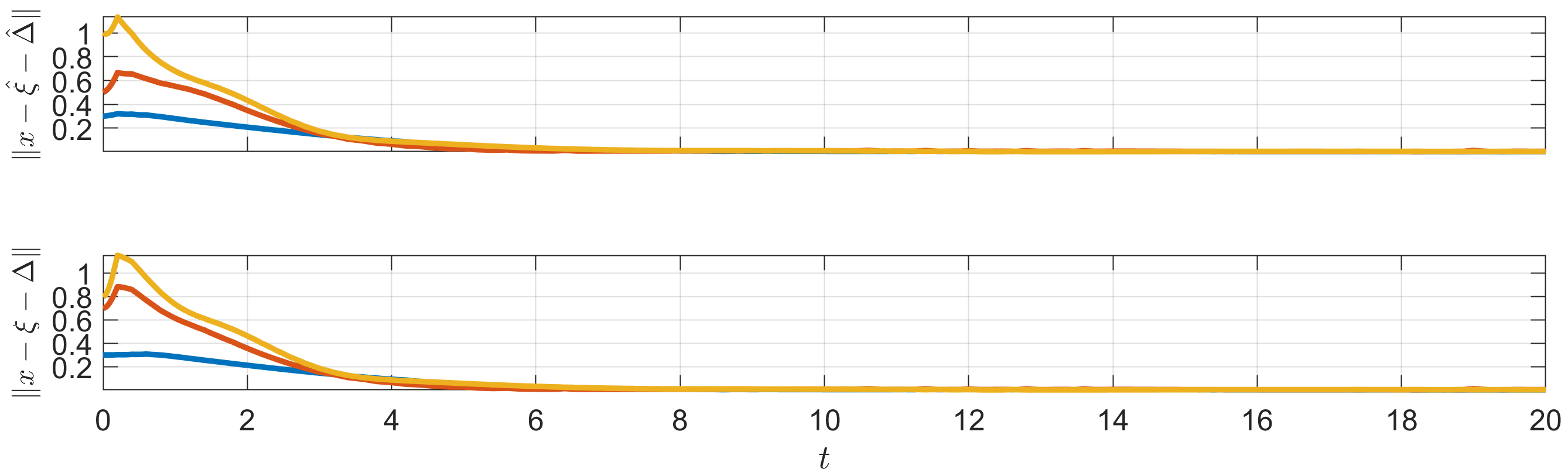}
		\caption{Norms of formation tracking errors}
		\label{6fig: EX1_e}
	\end{figure}
	\begin{figure}[t]
		\centering
		\includegraphics[width=\columnwidth]{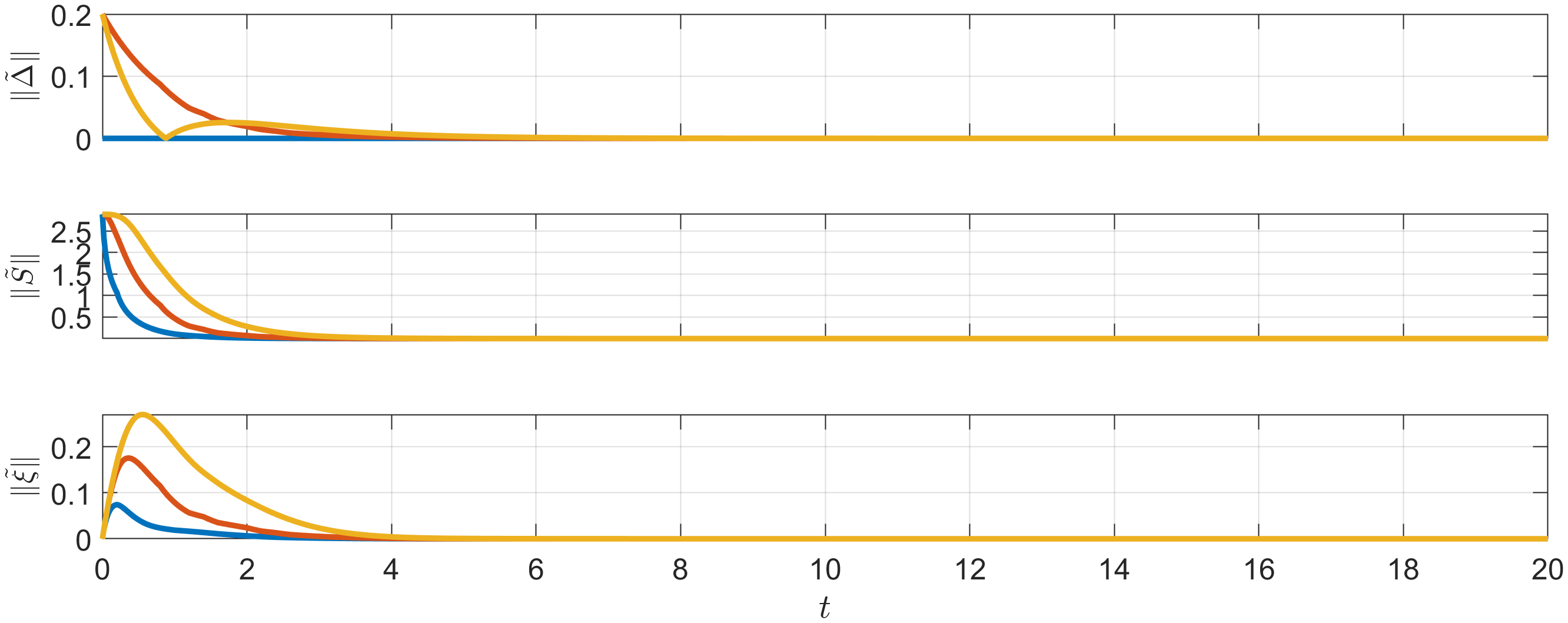}
		\caption{Norms of estimation errors}
		\label{6fig: EX1_o}
	\end{figure}
	\begin{figure}[t]
		\centering
		\includegraphics[width=\columnwidth]{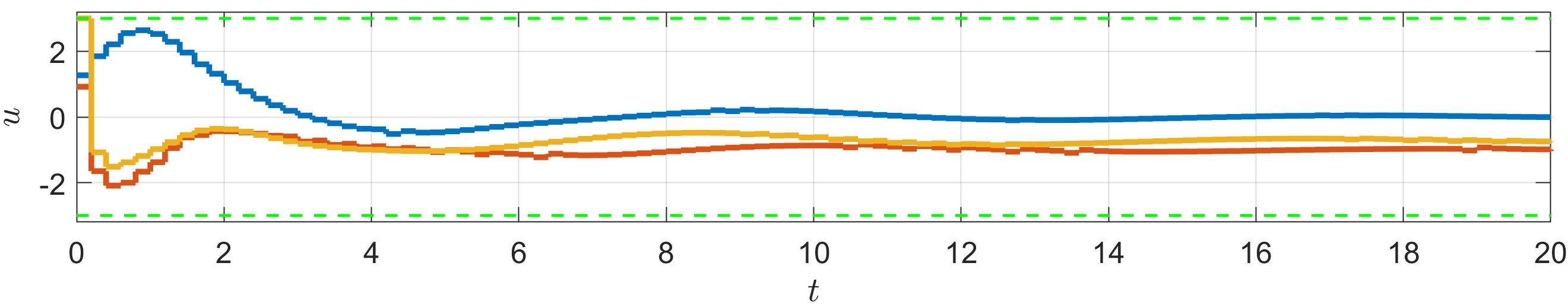}
		\caption{Control commands}
		\label{6fig: EX1_u}
	\end{figure}
	
	Let the dynamics of the leader node be
	\begin{align}
		\left\{
		\begin{array}{rl}
			\dot{\xi}_{0,1}&\!\!=\xi_{0,2}\\
			\dot{\xi}_{0,2}&\!\!=\xi_{0,3}\\
			\dot{\xi}_{0,3}&\!\!= -{\xi}_{0,1} -1.16\xi_{0,2} -2\xi_{0,3}
		\end{array}
		\right.
	\end{align}
	with $\xi_0(0)=\left[1\ 0\ 0\right]^\top$. 
	
	In simulations, the desired formation displacement vectors of the 3 followers with respect to the leader $0$ are set as ${\varDelta}_{10}=[0\ 0\ 0]^\top$, ${\varDelta}_{20}=[0.2\ 0\ 0]^\top$, ${\varDelta}_{30}=[-0.2\ 0\ 0]^\top$, ${\varDelta}_{40}=[2\ 0\ 0]^\top$. Time-varying edge weights, including the adjacency matrix and pinning gains, are designed to mimic faults in the communication network. In particular, the adjacency matrix and the pinning matrix are
	\begin{align}
		\mathcal{A}&=\left[\begin{array}{lll}
			0&0&\hspace{1cm}0\\
			1+0.5\sin(t)*{\rm rand([0,1])}&0&\hspace{1cm}0\\
			1&0&\hspace{1cm}0
		\end{array}\right]\\
		\mathcal{B}&=\left[\begin{array}{lll}
			1+0.3\sin(t)*{\rm rand([0,1])}&0&\hspace{1cm}0\\
			0&0&\hspace{1cm}0\\
			0&0&\hspace{1cm}0
		\end{array}\right]
	\end{align}
	where ${\rm rand}([0,1])$ is a random signal chosen from the interval $[0,1]$.

	The control parameters are selected following the obtained stability conditions. The sampling period for updating the control actions is set as $0.2$s. The leader state observation gains are chosen as $c_{\xi_1}=c_{\xi_2}=c_{\xi_3}=2$. In the definition of the sliding mode tracking error, $\lambda_{1,0}=\lambda_{2,0}=\lambda_{3,0}=1$ and $\lambda_{1,1}=\lambda_{2,1}=\lambda_{3,1}=2$. The control gains are $c_1=c_2=c_3=2$. In the MPC problems, the prediction horizon is 0.8s, $Q_1=Q_2=Q_3=10$ and $R_1=R_2=R_3=0.1$.
	
	The time-varying communication fault parameters introduced to the network are depicted in Figure \ref{6fig: EX1_th}. The simulation results, as shown in Figures \ref{6fig: EX1_y}-\ref{6fig: EX1_u}, illustrate the responses of the three followers with solid lines in blue, orange, and yellow, while the virtual leader's responses are represented with gray dashed lines.  Specifically, Figure \ref{6fig: EX1_y} illustrates the output trajectories of the followers, demonstrating that the formation tracking objective has been successfully achieved. The norms of the tracking errors, relative to both the estimated and actual leader states, are displayed in Figure \ref{6fig: EX1_e}. Additionally, Figure \ref{6fig: EX1_o} presents the estimation errors of the relative position displacement, the leader's dynamics, and the leader state. Figure \ref{6fig: EX1_u} illustrates the control commands applied to the followers, showing that the input constraint (\ref{6eq: EX1 input constraint}) is satisfied. 
	
	\subsection{Example 2: A Multi-UAV System}

	Next, we consider applying the proposed adaptive distributed control strategy to the outer-loop translation control of a group of UAVs. The networked UAV system comprises 5 UAVs, with their translational motions described by
	\begin{subequations}
		\begin{align}
			&	\left\{
			\begin{array}{rl}
				\dot{{\zeta}}_{i}&={v}_{i}\\
				\dot{{v}}_{i}&=-{g}+\frac{1}{m_{i}}{r}(u_{\eta_{i}})u_{F_{i}}
			\end{array}
			\right.
		\end{align}
	\end{subequations}
	where $i=1,2,\cdots,5$; $g=[0\ 0\ 9.81]^\top$ and $m_{i}=2.618
	$ kg; ${\zeta}_{i}\in\mathbb{R}^3$ and $v_{i}\in\mathbb{R}^3$ are position and velocity vectors of the 5 UAVs;  $u_{\eta_{i}}=\left[u_{\phi_{i}}\ u_{\theta_{i}}\ u_{\psi_{i}}\right]^\top$ and $u_{F_{i}}\in\mathbb{R}$ are control inputs of the translational subsystem, representing the desired rotation angles and the total thrust force, respectively. The control inputs suffer from the following input constraints: 
	\begin{align}
		\Omega_{u_{\eta_{i}}}&=\left\{u_{\eta_{i}}|[-0.5\ -0.5\ -0.1]^\top\leqslant u_{\eta_{i}} \leqslant [0.5\ 0.5\ 0.1]^\top \right\}\\
		\Omega_{u_{F_{i}}}&=\left\{u_{F_{i}}|-4\leqslant u_{F_{i}} \leqslant 4 \right\}
	\end{align}  
	The initial positions of the 5 UAVs are $\zeta_1(0)=[10\ 0\ 0]^\top$, $x_2(0)=[7\ 0\ 0]^\top$, $x_3(0)=[13\ 0\ 0]^\top$, $x_4(0)=[8.5\ 0\ 0]^\top$, $x_5(0)=[11.5\ 0\ 0]^\top$. Their initial linear velocities are all zero.

	Let the dynamics of the leader node be
	\begin{align}
		\dot{\xi}_{0}&\!=\!\left[\begin{smallmatrix}
			0     & 0    &  0    &  1 &    0    &  0\\
			0   &   0   &   0  &    0  &    1   &   0\\
			0  &    0    & 0 &     0   &   0&      1\\
			-0.0676 &0    &  0   &  -0.1040 &0    &  0\\
			0     &-0.0676& 0     & 0  &   -0.1040 &0\\
			0   &   0   &  -0.0025& 0  &    0   &  -0.02
		\end{smallmatrix}\right]\xi_{0}\!\!
	\end{align}
	with $\xi_0(0)=\left[10\ 0\ 0\ 0\ 3\ 1.2\right]^\top$.

	\begin{figure}[t]
		\centering
		\includegraphics[width=0.3\columnwidth]{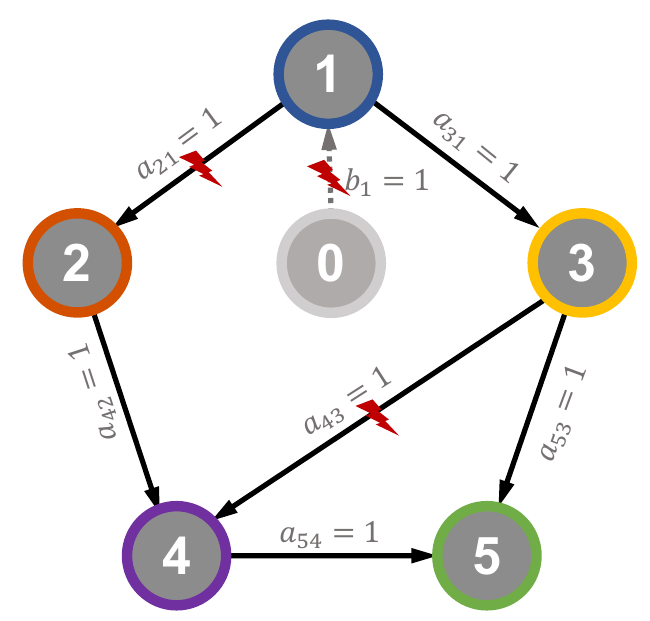}
		\caption{Formation shape and communication graph of the 5-UAV system}
		\label{6fig: EX2_g}
	\end{figure}
	\begin{figure}[t]
		\centering
		\includegraphics[width=\columnwidth]{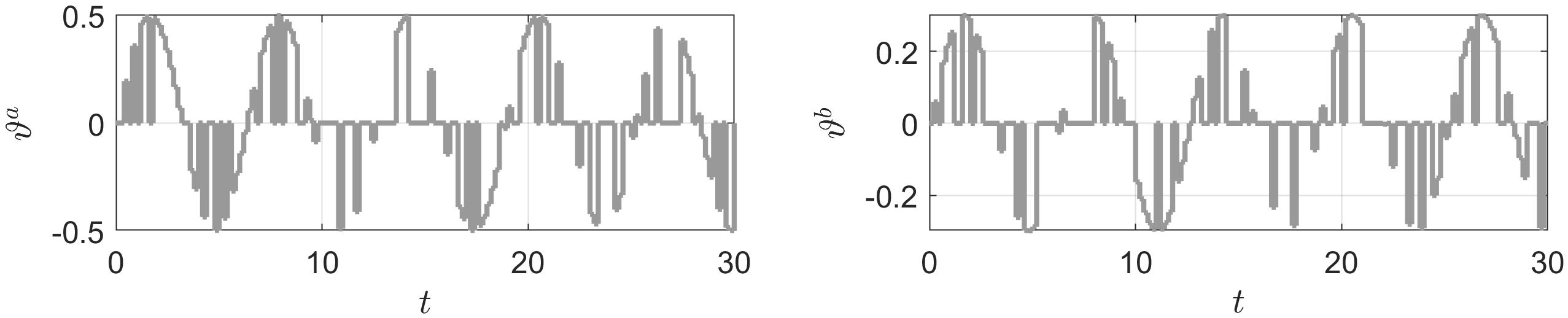}
		\caption{Communication link faults in the 5-UAV system}
		\label{6fig: EX2_th}
	\end{figure}
	\begin{figure}[t]
		\centering
		\includegraphics[width=\columnwidth]{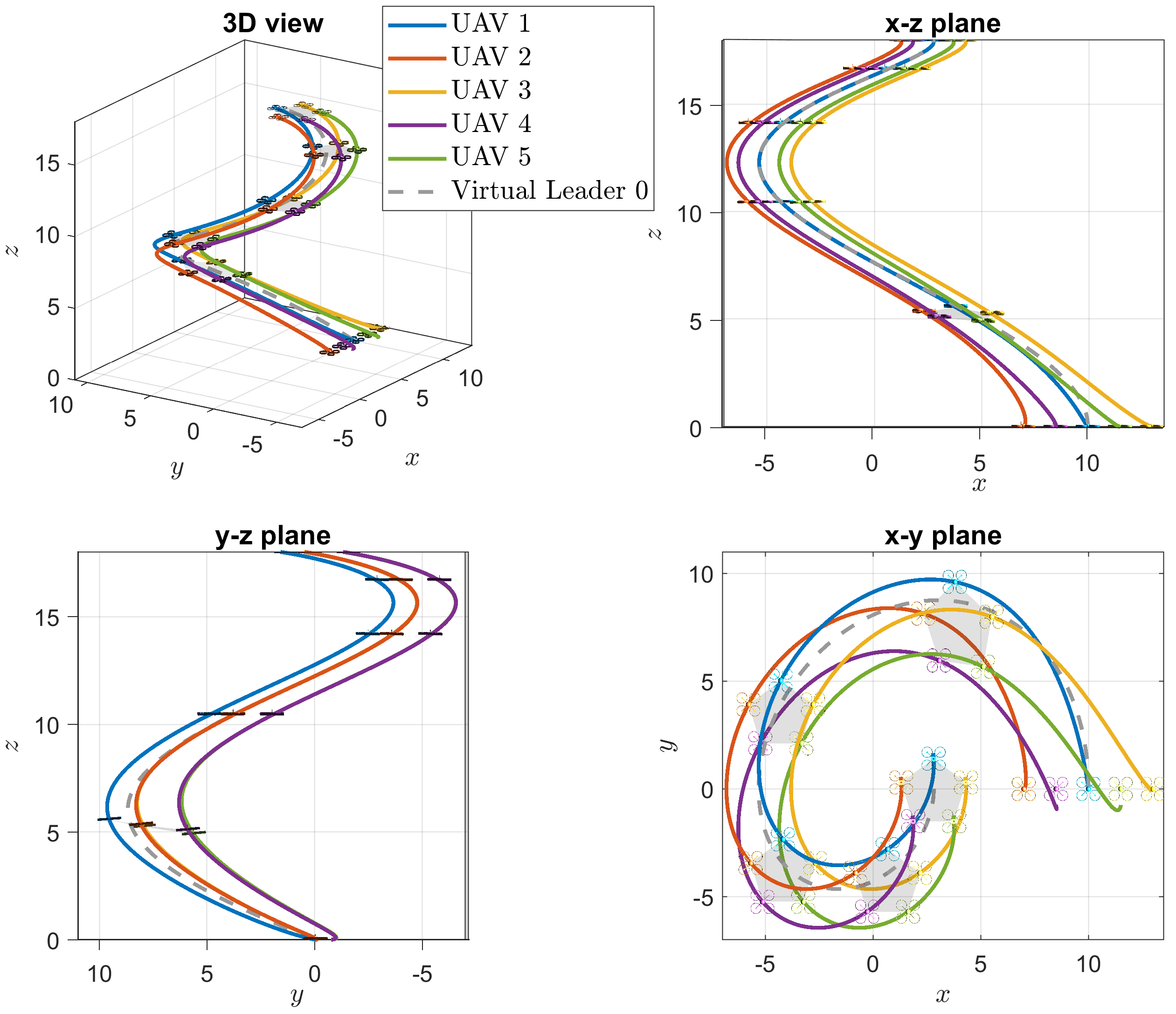}
		\caption{Formation tracking performance of the 5 UAVs}
		\label{6fig: EX2_y}
	\end{figure}
	\begin{figure}[t]
		\centering
		\includegraphics[width=\columnwidth]{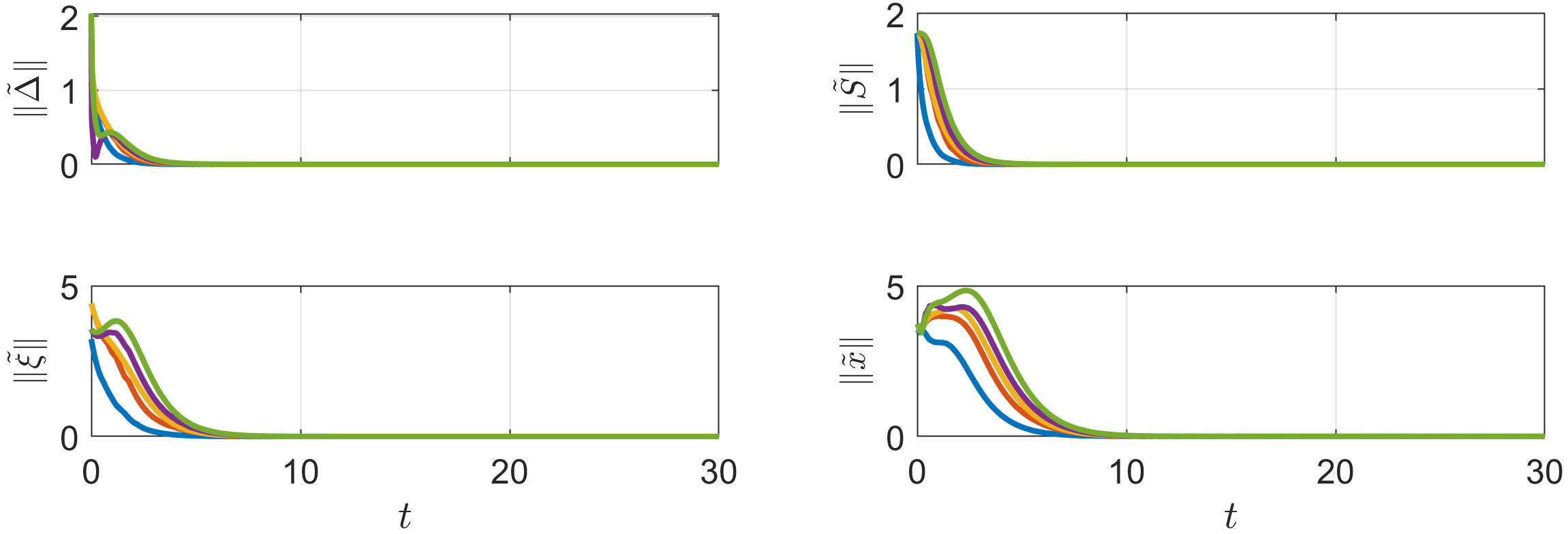}
		\caption{Norms of estimation and tracking errors of the 5 UAVs}
		\label{6fig: EX2_e}
	\end{figure}
	\begin{figure}[t]
		\centering
		\includegraphics[width=\columnwidth]{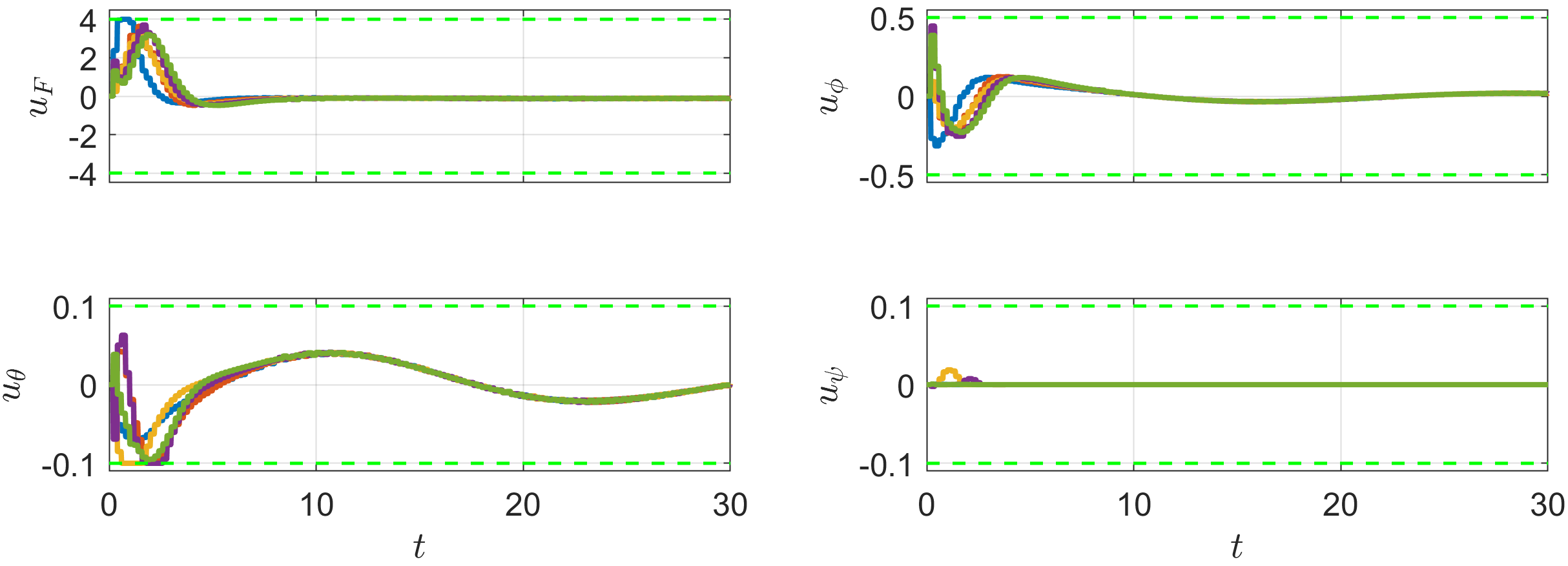}
		\caption{Control commands of the 5 UAVs}
		\label{6fig: EX2_u}
	\end{figure}
	
	The prescribed formation geometric shape and directed communication graph of the 5-UAV systems are illustrated in Figure \ref{6fig: EX2_g}. The desired displacement vectors of the followers with respect to the leader $0$ are set as ${\varDelta}_{10}=[0\ 1.1\ 0\ 0\ 0\ 0]^\top$, ${\varDelta}_{20}=[-1.5\ 0\ 0\ 0\ 0\ 0]^\top$, ${\varDelta}_{30}=[1.5\ 0\ 0\ 0\ 0\ 0]^\top$, ${\varDelta}_{40}=[-0.95\ -1.8\ 0\ 0\ 0\ 0]^\top$, ${\varDelta}_{50}=[0.95\ -1.8\ 0\ 0\ 0\ 0]^\top$. Time-varying edge weights, including the adjacency matrix and pinning gains, are designed to mimic faults in the communication network. In particular, the adjacency matrix and the pinning matrix are
	\begin{small}
		\begin{align}
			\mathcal{A}&=\left[\begin{smallmatrix}
				0&0&0&0&0\\
				1+0.5\sin(t) * {\rm rand([0,1])}&0&0&0&0\\
				1&0&0&0&0\\
				0&1&1+0.5\sin(t)*{\rm rand([0,1])}&0&0\\
				0&0&1&1&0
			\end{smallmatrix}\right]\\
			\mathcal{B}&={\scriptstyle\left[\begin{smallmatrix}
					1+0.3\sin(t)*{\rm rand([0,1])}&0&0&0&0\\
					0&0&0&0&0\\
					0&0&0&0&0
				\end{smallmatrix}\right]}
		\end{align}
	\end{small}
	where ${\rm rand}([0,1])$ is a random signal chosen from the interval $[0,1]$. The time-varying fault communication parameters added to the network are illustrated in Figure \ref{6fig: EX1_th}. 
	
	The control parameters are selected following the previously established stability conditions. The sampling period for updating the control actions is set as $0.2$s. The leader state observation gains are chosen as $c_{\xi_1}=c_{\xi_2}=c_{\xi_3}=c_{\xi_4}=c_{\xi_5}=1.2$. In the definition of the sliding mode tracking error, $\lambda_{1,0}=\lambda_{2,0}=\lambda_{3,0}=\lambda_{4,0}=\lambda_{5,0}=1$. The control gains are $c_1=c_2=c_3=c_4=c_5=2$. In the MPC problems, the prediction horizon is 0.8s, $Q_1=Q_2=Q_3=Q_4=Q_5=\text{diag}(2,2,5)$ and $R_1=R_2=R_3=R_4=R_5=\text{diag}(0.1,10,10,10)$. 
	
	The simulation results are illustrated in Figures \ref{6fig: EX2_y}-\ref{6fig: EX2_u}, where the responses of the 5 UAVs are depicted with solid lines in blue, orange, yellow, purple, and green, and the virtual leader's responses are depicted with gray dashed lines. Figure \ref{6fig: EX2_y} illustrates the formation tracking performance of the 5-UAV system in 3D and three-plane views. The norms of the estimation and tracking errors of the 5 UAVs are displayed in Figure \ref{6fig: EX2_e}, respectively. Figure \ref{6fig: EX2_u} shows the control commands of the 5 UAVs. It can be seen that their input constraints are all satisfied. 
	
	\section{Conclusions}\label{6s: Conclusions}
	A novel adaptive distributed observer-based DMPC method has been introduced in this paper, which is developed for nonlinear multi-agent formation tracking with input constraints and unknown communication faults. The method utilizes adaptive distributed observers in each local control system to estimate the state, dynamics, and relative position of the leader, enabling each agent to independently achieve formation tracking without direct access to the leader's information. The designed distributed MPC controllers use the estimated information to manipulate agents into a predefined formation while respecting input constraints. This research employs adaptive observers to reduce the complexity in the DMPC design, allowing for effective local controller formulation and resilient distributed formation tracking. 
	
	\bibliographystyle{IEEEtran}
	\bibliography{Ref}

\begin{thebibliography}{10}
\providecommand{\url}[1]{#1}
\csname url@samestyle\endcsname
\providecommand{\newblock}{\relax}
\providecommand{\bibinfo}[2]{#2}
\providecommand{\BIBentrySTDinterwordspacing}{\spaceskip=0pt\relax}
\providecommand{\BIBentryALTinterwordstretchfactor}{4}
\providecommand{\BIBentryALTinterwordspacing}{\spaceskip=\fontdimen2\font plus
\BIBentryALTinterwordstretchfactor\fontdimen3\font minus
  \fontdimen4\font\relax}
\providecommand{\BIBforeignlanguage}[2]{{%
\expandafter\ifx\csname l@#1\endcsname\relax
\typeout{** WARNING: IEEEtran.bst: No hyphenation pattern has been}%
\typeout{** loaded for the language `#1'. Using the pattern for}%
\typeout{** the default language instead.}%
\else
\language=\csname l@#1\endcsname
\fi
#2}}
\providecommand{\BIBdecl}{\relax}
\BIBdecl

\bibitem{ren2011distributed}
W.~Ren and Y.~Cao, \emph{Distributed Coordination of Multi-Agent Networks:
  Emergent Problems, Models, and Issues}.\hskip 1em plus 0.5em minus
  0.4em\relax Springer, 2011, vol.~1.

\bibitem{qin2016recent}
J.~Qin, Q.~Ma, Y.~Shi, and L.~Wang, ``Recent advances in consensus of
  multi-agent systems: A brief survey,'' \emph{IEEE Transactions on Industrial
  Electronics}, vol.~64, no.~6, pp. 4972--4983, 2016.

\bibitem{shi2021advanced}
Y.~Shi and K.~Zhang, ``Advanced model predictive control framework for
  autonomous intelligent mechatronic systems: A tutorial overview and
  perspectives,'' \emph{Annual Reviews in Control}, vol.~52, pp. 170--196,
  2021.

\bibitem{keviczky2004study}
T.~Keviczky, F.~Borrelli, and G.~J. Balas, ``A study on decentralized receding
  horizon control for decoupled systems,'' in \emph{Proceedings of the 2004
  American Control Conference}, vol.~6, Boston, MA, USA, May 2004, pp.
  4921--4926.

\bibitem{christofides2013distributed}
P.~D. Christofides, R.~Scattolini, D.~M. de~la Pena, and J.~Liu, ``Distributed
  model predictive control: A tutorial review and future research directions,''
  \emph{Computers \& Chemical Engineering}, vol.~51, pp. 21--41, 2013.

\bibitem{negenborn2014distributed}
R.~R. Negenborn and J.~M. Maestre, ``Distributed model predictive control: An
  overview and roadmap of future research opportunities,'' \emph{IEEE Control
  Systems Magazine}, vol.~34, no.~4, pp. 87--97, 2014.

\bibitem{dunbar2006distributed}
W.~B. Dunbar and R.~M. Murray, ``Distributed receding horizon control for
  multi-vehicle formation stabilization,'' \emph{Automatica}, vol.~42, no.~4,
  pp. 549--558, 2006.

\bibitem{li2013robust}
H.~Li and Y.~Shi, ``Robust distributed model predictive control of constrained
  continuous-time nonlinear systems: A robustness constraint approach,''
  \emph{IEEE Transactions on Automatic Control}, vol.~59, no.~6, pp.
  1673--1678, 2013.

\bibitem{zheng2016distributed}
Y.~Zheng, S.~E. Li, K.~Li, F.~Borrelli, and J.~K. Hedrick, ``Distributed model
  predictive control for heterogeneous vehicle platoons under unidirectional
  topologies,'' \emph{IEEE Transactions on Control Systems Technology},
  vol.~25, no.~3, pp. 899--910, 2016.

\bibitem{mercangoz2007distributed}
M.~Mercang{\"o}z and F.~J. Doyle~III, ``Distributed model predictive control of
  an experimental four-tank system,'' \emph{Journal of Process Control},
  vol.~17, no.~3, pp. 297--308, 2007.

\bibitem{stewart2010cooperative}
B.~T. Stewart, A.~N. Venkat, J.~B. Rawlings, S.~J. Wright, and G.~Pannocchia,
  ``Cooperative distributed model predictive control,'' \emph{Systems \&
  Control Letters}, vol.~59, no.~8, pp. 460--469, 2010.

\bibitem{venkat2005stability}
A.~N. Venkat, J.~B. Rawlings, and S.~J. Wright, ``Stability and optimality of
  distributed model predictive control,'' in \emph{Proceedings of the 44th IEEE
  Conference on Decision and Control}, Seville, Spain, Dec. 2005, pp.
  6680--6685.

\bibitem{richards2004decentralized}
A.~Richards and J.~How, ``A decentralized algorithm for robust constrained
  model predictive control,'' in \emph{Proceedings of the 2004 American Control
  Conference}, vol.~5, Boston, MA, USA, Jun. 2004, pp. 4261--4266.

\bibitem{richards2007robust}
A.~Richards and J.~P. How, ``Robust distributed model predictive control,''
  \emph{International Journal of Control}, vol.~80, no.~9, pp. 1517--1531,
  2007.

\bibitem{li2019robust}
Z.~Li and J.~Chen, ``Robust consensus for multi-agent systems communicating
  over stochastic uncertain networks,'' \emph{SIAM Journal on Control and
  Optimization}, vol.~57, no.~5, pp. 3553--3570, 2019.

\bibitem{li2014multi}
T.~Li, F.~Wu, and J.-F. Zhang, ``Multi-agent consensus with
  relative-state-dependent measurement noises,'' \emph{IEEE Transactions on
  Automatic Control}, vol.~59, no.~9, pp. 2463--2468, 2014.

\bibitem{ma2015mean}
X.~Ma and N.~Elia, ``Mean square performance and robust yet fragile nature of
  torus networked average consensus,'' \emph{IEEE Transactions on Control of
  Network Systems}, vol.~2, no.~3, pp. 216--225, 2015.

\bibitem{wang2010consensus}
J.~Wang and N.~Elia, ``Consensus over networks with dynamic channels,''
  \emph{International Journal of Systems, Control and Communications}, vol.~2,
  no. 1-3, pp. 275--297, 2010.

\bibitem{zelazo2015robustness}
D.~Zelazo and M.~B{\"u}rger, ``On the robustness of uncertain consensus
  networks,'' \emph{IEEE Transactions on Control of Network Systems}, vol.~4,
  no.~2, pp. 170--178, 2015.

\bibitem{chen2020adaptive}
C.~Chen, K.~Xie, F.~L. Lewis, S.~Xie, and R.~Fierro, ``Adaptive synchronization
  of multi-agent systems with resilience to communication link faults,''
  \emph{Automatica}, vol. 111, p. 108636, 2020.

\bibitem{yang2021adaptive}
Q.~Yang, Y.~Lyu, X.~Li, C.~Chen, and F.~L. Lewis, ``Adaptive distributed
  synchronization of heterogeneous multi-agent systems over directed graphs
  with time-varying edge weights,'' \emph{Journal of the Franklin Institute},
  vol. 358, no.~4, pp. 2434--2452, 2021.

\bibitem{dunbar2007distributed}
W.~B. Dunbar, ``Distributed receding horizon control of dynamically coupled
  nonlinear systems,'' \emph{IEEE Transactions on Automatic Control}, vol.~52,
  no.~7, pp. 1249--1263, 2007.

\bibitem{wei2024robust}
H.~Wei, C.~Liu, and Y.~Shi, ``A robust distributed {MPC} framework for
  multi-agent consensus with communication delays,'' \emph{IEEE Transactions on
  Automatic Control}, 2024.

\bibitem{wei2021robust}
H.~Wei, Q.~Sun, J.~Chen, and Y.~Shi, ``Robust distributed model predictive
  platooning control for heterogeneous autonomous surface vehicles,''
  \emph{Control Engineering Practice}, vol. 107, p. 104655, 2021.

\bibitem{wei2019distributed}
H.~Wei, C.~Shen, and Y.~Shi, ``Distributed {L}yapunov-based model predictive
  formation tracking control for autonomous underwater vehicles subject to
  disturbances,'' \emph{IEEE Transactions on Systems, Man, and Cybernetics:
  Systems}, vol.~51, no.~8, pp. 5198--5208, 2019.

\end{thebibliography}
	
\end{document}